\documentclass[journal]{IEEEtran}
\usepackage{amsmath,amssymb,amsthm}
\usepackage{algorithm}
\usepackage{algorithmic}
\usepackage{graphicx}
\usepackage{cite}
\usepackage{hyperref}
\usepackage{bm}
\usepackage{bbm}
\usepackage{enumitem}
\usepackage{booktabs}

\newtheorem{theorem}{Theorem}
\newtheorem{lemma}{Lemma}
\newtheorem{proposition}{Proposition}
\newtheorem{corollary}{Corollary}
\newtheorem{definition}{Definition}
\newtheorem{remark}{Remark}
\newtheorem{assumption}{Assumption}

\begin{document}

\title{Heterogeneous Mean Field Game Framework for LEO Satellite-Assisted 
V2X Networks}

\author{Kangkang Sun,
        Jianhua Li,~\IEEEmembership{Senior Member,~IEEE,}
        Xiuzhen Chen,~\IEEEmembership{Member,~IEEE,}\\
        Mingzhe Chen,~\IEEEmembership{Senior Member,~IEEE,}
        Minyi Guo,~\IEEEmembership{Fellow,~IEEE}
  \thanks{Kangkang Sun, Jianhua Li, Xiuzhen Chen and Minyi Guo are
          with the Shanghai Key Laboratory of Integrated Administration
          Technologies for Information Security, School of Computer Science,
          Shanghai Jiao Tong University, Shanghai 200240, China
          (e-mail: szpsunkk@sjtu.edu.cn; lijh888@sjtu.edu.cn;
          xzchen@sjtu.edu.cn; guo-my@cs.sjtu.edu.cn).}
  \thanks{Mingzhe Chen is with the Department of Electrical and Computer Engineering, University of Miami, Coral Gables, FL 33146 USA (e-mail: mingzhe.chen@miami.edu).}
  \thanks{\textit{Corresponding author: Jianhua Li and Xiuzhen Chen}.}
  \thanks{This work has been submitted to the IEEE for possible publication. Copyright may be transferred without notice, after which this version may no longer be accessible.}
}

\maketitle

\begin{abstract}
Coordinating mixed fleets of massive vehicles under stringent delay constraints is a central scalability bottleneck in next-generation mobile computing networks, especially when passenger cars, freight trucks, and autonomous vehicles share the same radio and multi-access edge computing (MEC) infrastructure. Heterogeneous mean field games (HMFG) are a principled framework for this setting, but a fundamental design question remains open: how many agent types should be used for a fleet of size $N$? The difficulty is a two-sided trade-off that existing theory does not resolve: using more types improves heterogeneity representation, but it reduces per-class sample size and weakens the mean-field approximation accuracy. This paper resolves that trade-off through an explicit $\varepsilon$-Nash error decomposition, a closed-form type-selection law, a heterogeneity-aware equilibrium solver, and a robust extension to time-varying LEO backhaul dynamics. For the 1D queue state space, the optimal type count satisfies $K^*(N)=\Theta(N^{1/3})$; for the joint queue-channel model ($d=2$), the scaling becomes $K^*(N)=\Theta(N^{1/5})$ with logarithmic correction. The unified formula $K^*(N)=\Theta(N^{\alpha/(\alpha+\beta)})$ provides dimension-dependent design guidance, reducing type granularity to a principled, set-once system parameter rather than a per-deployment tuning burden. Experiments validate the 1D scaling law with empirical slope $0.334 \pm 0.004$, achieve $2.3\times$ faster PDHG convergence at $K=5$, and deliver up to $29.5\%$ lower delay and $60\%$ higher throughput than homogeneous baselines. Unlike model-free DRL methods whose training complexity scales with the state-action space, the proposed HMFG solver has per-iteration complexity $O(K^2 N_q N_t)$ independent of fleet size $N$, making it suitable for large-scale mobile edge computing deployment.
\end{abstract}


\section{Introduction}
\label{sec:intro}

The global proliferation of connected vehicles is reshaping the design requirements of next-generation wireless networks. By 2030, metropolitan areas are projected to support fleets of $10^4$ to $10^5$ heterogeneous vehicles, spanning passenger cars, freight trucks, and autonomous vehicles, each generating continuous sensor traffic that must be offloaded to multi-access edge computing (MEC) servers under strict latency constraints \cite{lasry2007mean,huang2006large}. In large-scale vehicle-to-everything (V2X) systems, the core resource-allocation question is not only how to allocate power and computation, but whether the underlying game model faithfully represents the population it coordinates.

When $N$ vehicles interact directly, coordination complexity grows as $O(N^2)$, making exact Nash computation intractable for even moderate fleet sizes. Mean field game (MFG) theory \cite{lasry2007mean,huang2006large} resolves this by replacing the aggregate influence of all other vehicles with a population-level distribution, reducing per-iteration complexity to $O(1)$ in vehicle count. Building on this, Kang et al.\ \cite{kang2023time} developed a G-prox Primal-Dual Hybrid Gradient (PDHG) solver for automotive MEC offloading with vehicle-count-independent behavior, and Wang et al.\ \cite{wang2023distributed} extended MFG to LEO-assisted offloading. These works demonstrate that MFG is practically viable for large-scale V2X coordination.

Homogeneous MFG treats all vehicles as statistically identical, which is unrealistic for mixed fleets. Autonomous vehicles may generate substantially higher sensor loads than passenger cars, while freight vehicles operate under distinct energy and delay constraints. Ignoring such heterogeneity inflates the $\varepsilon$-Nash approximation error and degrades fairness and throughput across classes. Heterogeneous MFG (HMFG) addresses this by partitioning the fleet into $K$ agent types \cite{zhang2022hmfmarl,xu2025joint,qiao2025hmfg}. Empirical studies report significant gains from type-aware control, confirming the practical value of heterogeneity modeling.

Despite this progress, one key question remains unresolved: for a fleet of $N$ vehicles, how many types $K$ should be distinguished? Existing HMFG/MARL works often fix $K$ heuristically \cite{zhang2022hmfmarl,xu2025joint}, while the well-posedness analysis in \cite{qiao2025hmfg} does not provide an error-minimizing selection rule. The difficulty comes from two opposing effects. Increasing $K$ decreases type-discretization error because finer bins better represent population diversity. At the same time, increasing $K$ reduces per-class sample size $N_k\approx N/K$, weakening the finite-sample accuracy of the mean-field approximation and increasing sampling error. The balance yields a nontrivial optimal $K^*(N)$ not determined by classical quantization theory \cite{gray1998quantization} or graphon MFG \cite{caines2021graphon,carmona2022graphon} alone. The challenge is sharpened in LEO-assisted V2X. Fast satellite topology changes introduce non-stationary backhaul surcharges in HMFG dynamics, potentially affecting the stability and robustness of iterative equilibrium solvers designed under static assumptions \cite{wang2023distributed,guo2024semcom}.

Our approach is to make the type-granularity question the organizing principle of the entire framework. We first derive an explicit decomposition of the HMFG approximation error into a heterogeneity-representation term and a finite-sample mean-field term. That decomposition yields a closed-form scaling law for the type count, guides a heterogeneity-aware PDHG step-size rule, and extends naturally to LEO temporal-graph backhaul dynamics through an additive robustness analysis. This design makes the three questions above analytically connected rather than separate engineering heuristics.

A rigorous framework is established for optimal type granularity in HMFG-based V2X networks with LEO-assisted backhaul. The interplay among type granularity, solver stability, and network dynamics gives rise to three closely related questions whose answers determine whether HMFG can be deployed at scale:
\begin{itemize}[leftmargin=*]
    \item \textit{\textbf{Q1 (Granularity):} For a fleet of $N$ vehicles, what is the 
    optimal type count $K^*(N)$, and how does it scale with $N$?}
    \item \textit{\textbf{Q2 (Convergence):} How should PDHG step sizes adapt to 
    heterogeneity to guarantee stable and fast convergence?}
    \item \textit{\textbf{Q3 (Robustness):} Does the optimal granularity law remain 
    valid under time-varying LEO topology?}
\end{itemize}
The answers are concise and actionable: Q1 yields $K^*(N)=\Theta(N^{1/3})$, Q2 yields a heterogeneity-dependent step-size shrinkage rule, and Q3 is affirmative under a mild topology-variation condition, summarized next as three technical contributions.

\begin{itemize}[leftmargin=*]
    \item \textit{Optimal type-count law.} We derive an explicit error decomposition and 
    show that the unique error-minimizing type count satisfies 
    $K^*(N)=\Theta(N^{1/3})$ in the 1D queue setting. The resulting take-home message is 
    that even a $10^5$-vehicle fleet needs only about 28 types.
    \item \textit{Heterogeneity-aware convergence.} We establish a sufficient multi-type 
    PDHG step-size condition and design an adaptive rule that stabilizes heterogeneous 
    coupling while accelerating convergence, yielding a $2.3\times$ speedup at $K=5$.
    \item \textit{LEO-robust scalable framework.} We embed temporal-graph LEO backhaul 
    dynamics into HMFG and show that cube-root scaling remains order-optimal under bounded 
    topology variation, while the full pipeline improves delay and throughput by up to 
    $29.5\%$ and $60\%$, respectively, over homogeneous baselines.
\end{itemize}

Relative to \cite{kang2023time}, the proposed framework reduces to the homogeneous case at $K=1$, and the convergence condition recovers the classical criterion when $H_K=0$. Relative to \cite{zhang2022hmfmarl,xu2025joint}, a principled type-count selection rule is provided in place of a fixed heuristic choice of $K$. Relative to \cite{qiao2025hmfg}, the practical decomposition needed to derive explicit $K^*$ to $N$ scaling is sharpened. Relative to \cite{guo2024semcom}, temporal-graph LEO dynamics are embedded directly into HMFG and robustness guarantees are established. Relative to \cite{ding2026uav}, the analytical HMFG guarantees complement optimization-driven learning in uncertain and non-stationary regimes.

\begin{remark}[On Theoretical Novelty]
While the individual mathematical tools---Wasserstein convergence rates \cite{fournier2015rate}, Gr\"{o}nwall-type stability \cite{qiao2025hmfg}, Chambolle--Pock convergence \cite{chambolle2011first}---are established, their composition into a unified type-granularity framework with an explicit and actionable scaling law is new. The practical impact is that a continuous design space is reduced to a single closed-form formula, making type granularity a principled system-design choice rather than a tuning hyperparameter.
\end{remark}

The remainder of the paper is organized as follows. Section~\ref{sec:related} reviews related work. Section~\ref{sec:model} presents the V2X system model and HMFG formulation. Sections~\ref{sec:het_measure} and \ref{sec:optimal_K} develop the error decomposition and optimal granularity law. Section~\ref{sec:algorithm} presents the three-algorithm framework, convergence analysis, and online deployment considerations. Section~\ref{sec:simulation} reports numerical evaluation. Section~\ref{sec:conclusion} concludes.

\textit{Notation.} $\mathcal{P}(E)$ denotes the space of Borel probability measures on a Polish space $E$. $W_p(\mu,\nu)$ is the Wasserstein-$p$ distance. $\|\cdot\|_{L^2}$ and $\|\cdot\|_{H^1}$ are standard Sobolev norms. For $n\in\mathbb{N}$, $[n]:=\{1,\ldots,n\}$. $\mathbbm{1}_A$ is the indicator of event $A$. Bold symbols $\boldsymbol{\rho}=(\rho^{(1)},\ldots,\rho^{(K)})$ denote $K$-tuples; $\|\mathbf{A}\|_{\mathrm{op}}$ is the operator norm.

\section{Related Work}
\label{sec:related}

Prior work falls into three strands: LEO-assisted V2X resource management, 
which motivates the \emph{problem setting}; mean field game theory with 
heterogeneous extensions and solvers, which motivates the \emph{algorithmic 
framework}; and complementary approaches from quantization, graphon limits, 
and optimization-driven learning.

\subsection{LEO Satellite-Assisted V2X}
\label{subsec:related_satv2x}

Surveys \cite{NoorRahim2022,Nair2024} identify fixed RSU/cellular backhaul 
as the primary bottleneck in highway and rural deployments, motivating LEO 
augmentation \cite{Rajalakshmi2024}. Wang et al.\ \cite{wang2023distributed,wang2022smfg_conf} 
applied Stackelberg MFG to distributed offloading in LEO networks; 
Guo et al.\ \cite{guo2024semcom} introduced the \emph{temporal-graph snapshot 
model} for Starlink-based constellations (60\,s per snapshot); Kang et al.\ 
\cite{kang2024satellite_fl} addressed satellite federated learning under orbital 
dynamics; and Chen et al.\ \cite{Chen2025a,Chen2025b} analyzed 
game-theoretic offloading in LEO-terrestrial and UAV-LEO systems 
\cite{Tang2021,jiang2024survey,shen2025chimera}.

Despite this progress, existing works share two gaps: they assume homogeneous 
vehicle populations or fix $K$ heuristically, and none analyzes how LEO 
topology dynamics affect optimal HMFG granularity or solver convergence stability. 
Our framework addresses both by embedding temporal-graph dynamics into HMFG, 
proving $K^*(N)=\Theta(N^{1/3})$ remains order-optimal under bounded topology 
variation (Theorem~\ref{thm:leo_robustness}), and providing heterogeneity-aware 
step-size adaptation for non-stationary backhaul.

\subsection{Mean Field Games: Heterogeneous Extensions and Solvers}
\label{subsec:related_mfg}

MFG theory \cite{lasry2007mean,huang2006large} enables scalable V2X resource 
management via HJB--FPK decoupling. Kang et al.\ \cite{kang2023time} showed 
that G-prox PDHG \cite{chambolle2011first} achieves vehicle-count-independent 
complexity for automotive MEC. Wang et al.\ \cite{wang2024mfg_iscc} extended 
MFG to sensing-communication-computation waveform design, while Zhang et al.\ 
\cite{Zhang2023} and Xu et al.\ \cite{Xu2024} integrated MFG with MARL for 
spectrum allocation and UAV-assisted V2X.

For heterogeneous fleets, Zhang et al.\ \cite{zhang2022hmfmarl} proposed 
two-type HMFG-MARL for SAGIN routing ($\sim$80\% throughput gain) and 
Xu et al.\ \cite{xu2025joint} developed multi-type MARL for joint V2X 
offloading. Theoretically, Qiao \cite{qiao2025hmfg} established local 
well-posedness for multi-population HMFG, while Carmona and Delarue 
\cite{carmona2018probabilistic} provide the probabilistic MFG apparatus; 
complementary results on Wasserstein-space operators \cite{Pham2023} and 
density-constrained heterogeneous settings \cite{Meszaros2022} align with 
our error decomposition but do not address optimal type granularity.

Across all these works, $K$ is fixed by design (typically $K \in \{1,2,3\}$) 
without a rule linking $K$ to fleet size or sampling accuracy. Our 
Theorem~\ref{thm:error_decomp} fills this gap: using the sharp 1D Wasserstein 
rate \cite{fournier2015rate}, the decomposition directly yields the cube-root 
law, and our Theorem~\ref{thm:convergence} provides the first 
heterogeneity-aware PDHG step-size rule.

\subsection{Complementary Approaches}
\label{subsec:related_compl}

Finite-type MFG design resembles vector quantization \cite{gray1998quantization}, 
but mean-field validity adds a sampling constraint absent in source coding. 
Graphon MFG \cite{caines2021graphon,carmona2022graphon,Caines2021,Hu2023,Fabian2023} 
uses a continuum of types but has no finite-sample analog; Weed and Bach 
\cite{weed2019sharp} sharpen empirical Wasserstein rates under low intrinsic 
dimension, consistent with our analysis. Neither framework jointly optimizes 
discretization fidelity and per-class sample reliability; our work bridges this gap.

Optimization-driven DRL \cite{ding2026uav} provides a complementary data-driven 
perspective with model-based acceleration. Our HMFG approach offers provable 
guarantees with $N$-independent per-iteration complexity; the two paradigms 
target different operating regimes and are potentially synergistic.

\section{System Model and HMFG Formulation}
\label{sec:model}

\subsection{Scenario Description}

As illustrated in Fig.~\ref{fig:v2x_background}, we consider a V2X 
communication network where a roadside infrastructure node (RSU) equipped 
with multi-access edge computing (MEC) serves a large mixed fleet of $N$ 
vehicles. The vehicles are heterogeneous in terms of data generation rates, 
energy constraints, and QoS requirements, which motivates a $K$-type 
partition framework for HMFG-based resource allocation.

\begin{figure}[!t]
\centering
\includegraphics[width=0.5\textwidth]{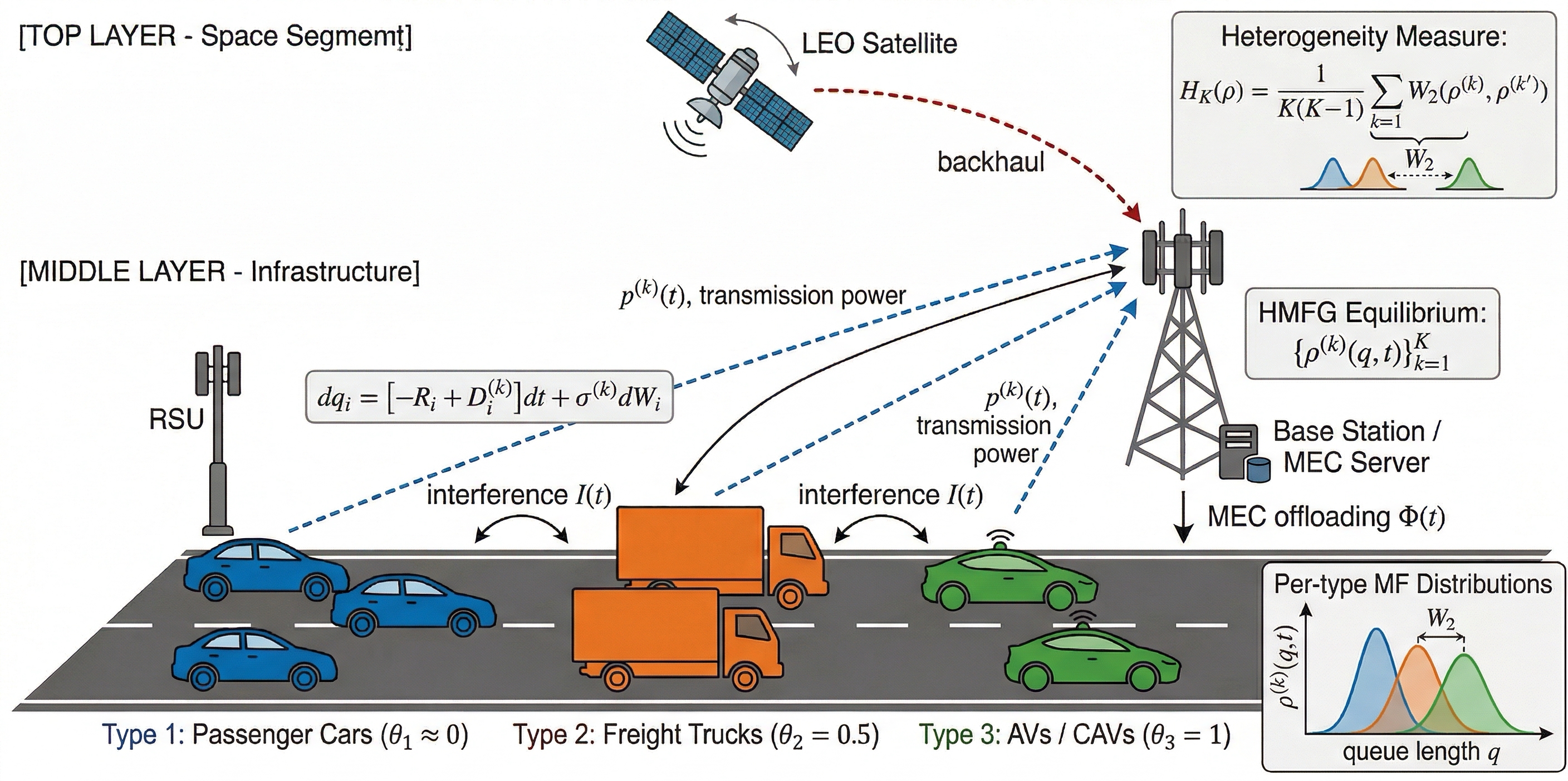}
\caption{Illustrative V2X--MEC scenario for heterogeneous mean field 
resource allocation. The diverse vehicle population is partitioned into 
$K$ types with distinct data loads and preferences. Per-type mean fields 
$\rho^{(k)}$ capture class-wise congestion, while cross-type coupling 
enters through interference $I(t)$ and the computational price $\Phi(t)$.}
\label{fig:v2x_background}
\end{figure}

Let $\mathcal{N} = \{1,\ldots,N\}$ denote the vehicle set, which is 
partitioned into $K$ disjoint types:
\begin{equation}
    \mathcal{N} = \bigsqcup_{k=1}^K \mathcal{C}_k, \quad 
    |\mathcal{C}_k| = N_k, \quad \sum_{k=1}^K N_k = N,
    \label{eq:partition}
\end{equation}
where type $k \in [K]$ represents a distinct vehicle class. For example, 
in a three-class model: $k=1$ for passenger cars, $k=2$ for freight trucks, 
and $k=3$ for autonomous vehicles. The heterogeneity across types manifests 
in three dimensions:
\begin{itemize}[leftmargin=*]
    \item \textit{State dynamics}: Data generation rates $D_i^{(k)}(t)$ 
    differ across vehicle classes, with autonomous vehicles generating 
    significantly more sensor data than passenger cars.
    \item \textit{Cost preferences}: Energy-vs.-delay trade-off weights 
    $\beta_1^{(k)}, \beta_2^{(k)}$ encode class-specific priorities.
    \item \textit{Action constraints}: Maximum transmission power and 
    computational budgets reflect hardware differences among vehicle types.
\end{itemize}

\subsection{State Dynamics}

The primary state of vehicle $i \in \mathcal{C}_k$ is its data queue length 
$q_i(t) \in [0, Q_{\max}]$. For the \emph{canonical 1D setting} analyzed in 
Sections~\ref{sec:het_measure}--\ref{sec:optimal_K}, the queue dynamics follow:
\begin{equation}
    dq_i(t) = \underbrace{\left[-R_i(t) + D_i^{(k)}(t)\right]}_{\text{drift}}dt 
    + \underbrace{\sigma^{(k)}\,dW_i(t)}_{\text{diffusion}},
    \label{eq:dynamics}
\end{equation}
where $R_i(t)$ is the transmission rate, $D_i^{(k)}(t)$ is the type-$k$ 
dependent data generation rate, $\sigma^{(k)} > 0$ captures queue uncertainty, 
and $\{W_i(t)\}$ are independent standard Brownian motions.

\begin{remark}[2D Joint Queue-Channel Extension]
\label{rem:2d_model}
For the \emph{joint queue-channel} model (Corollary~\ref{cor:scaling_2d}), 
the state is $x_i(t) = (q_i(t), g_i(t)) \in [0,Q_{\max}] \times [g_{\min}, g_{\max}]$, 
where the channel gain $g_i(t)$ follows an Ornstein-Uhlenbeck process:
\begin{equation}
    dg_i(t) = -\eta\bigl(g_i(t) - \bar{g}\bigr)\,dt + \sigma_g\,dB_i(t),
    \label{eq:channel_sde}
\end{equation}
with mean-reversion rate $\eta > 0$, mean gain $\bar{g}$, and independent 
Brownian motion $B_i(t)$. The 2D empirical Wasserstein rate 
$\alpha = 1/4$ (with logarithmic correction \cite{fournier2015rate}) applies, 
yielding $K^*(N) = \Theta(N^{1/5})$ per Corollary~\ref{cor:scaling_2d}. 
The theoretical derivations of Sections~\ref{sec:het_measure}--\ref{sec:optimal_K} 
apply verbatim to this 2D setting by substituting the appropriate $\alpha$. 
Numerical validation of the $N^{1/5}$ scaling in the 2D case, which requires 
a $50 \times 30$ joint state grid and increased computational budget, is 
deferred to a dedicated follow-on study and is consistent with the 
empirical dimension-dependent patterns reported in \cite{fournier2015rate}.
\end{remark} The transmission rate under the SINR model:
\begin{equation}
    R_i(t) = B\log_2\!\left(1 + \frac{p_i(t)g_i(t)}{\sigma_c^2 + I(t)}\right),
    \label{eq:rate}
\end{equation}
where $B$ is the channel bandwidth, $p_i(t) \in [0, P_{\max}]$ is the 
transmission power, $g_i(t)$ is the channel gain, $\sigma_c^2$ is the 
noise variance, and $I(t) = \sum_{j \neq i} p_j(t)g_j(t)$ is the 
aggregate interference.

\subsection{Cost Function Formulation}

Each vehicle $i$ of type $k$ minimizes the expected cumulative cost over 
the time horizon $[0,T]$:
\begin{equation}
    \begin{aligned}
        J_i^{(k)}(p_i) = &\mathbb{E}\!\left[\int_0^T 
    \left(\beta_1^{(k)} p_i^2(t) + \beta_2^{(k)}\Phi(t)R_i(t)\right)dt \right. \\
    &\left. + C^{(k)} q_i(T)^2\right],
    \end{aligned}
    \label{eq:cost}
\end{equation}
where $\beta_1^{(k)}, \beta_2^{(k)} > 0$ are type-dependent energy and 
offloading weights, $C^{(k)}$ is the terminal queue penalty, and 
\begin{equation}
    \Phi(t) = \kappa + \varrho \int_\Omega R(q,t)\rho(q,t)dq 
    + \frac{\mu}{B_{\mathrm{sat}}^{\tau(t)}},
    \label{eq:price_extended}
\end{equation}
is the congestion-aware computational price with LEO backhaul correction, where 
$B_{\mathrm{sat}}^{\tau(t)}$ is the available satellite backhaul bandwidth in the 
active snapshot $\tau(t)$ and $\mu>0$ is the backhaul-cost coefficient. The 
heterogeneous cost weights 
$\{\beta_1^{(k)}, \beta_2^{(k)}, C^{(k)}\}$ constitute the primary source 
of agent heterogeneity in our HMFG formulation.

\subsection{LEO Satellite-Assisted Backhaul Model}
\label{subsec:leo_model}

Following the temporal-graph view in \cite{guo2024semcom}, we represent the LEO 
constellation by snapshots $\mathcal{G}^{\tau}=(\mathcal{V}^{\tau},\mathcal{L}^{\tau})$ 
that remain approximately stable over each time window $\tau$. The effective RSU-MEC 
backhaul bandwidth through the current satellite route is modeled as
\begin{equation}
    B_{\mathrm{sat}}^{\tau} = \min_{(O_i,O_j)\in \mathcal{L}^{\tau}} r_{O_i,O_j}^{\tau},
    \label{eq:sat_bandwidth}
\end{equation}
where $r_{O_i,O_j}^{\tau}$ denotes the per-link capacity on the selected path.
This model captures the fact that fast topology changes perturb the backhaul 
term in \eqref{eq:price_extended}, thereby affecting both the HJB cost and FPK 
drift through coupled congestion dynamics.
Following the Starlink-oriented setting in \cite{guo2024semcom}, we use 
$r_{O_i,O_j}^{\tau}\sim\mathcal{U}[300,350]$ Mbps and $\Delta\tau=60$ s in 
simulations, which yields a bounded snapshot perturbation $\Delta_\Phi$ in 
\eqref{eq:leo_error} for realistic $\mu\le 1$.

\subsection{Heterogeneous Mean Field Game Formulation}

Following the HMFG framework of \cite{qiao2025hmfg}, as $N \to \infty$ 
with $K$ fixed, the finite-player game converges to a $K$-population MFG 
characterized by the per-type mean field densities 
$\boldsymbol{\rho} = \{\rho^{(k)}(q,t)\}_{k=1}^K$. The HMFG equilibrium 
is defined by the coupled Hamilton-Jacobi-Bellman (HJB) and 
Fokker-Planck-Kolmogorov (FPK) system:

\textbf{HJB equations} (backward in time, for each $k \in [K]$):
\begin{equation}
    -\frac{\partial V^{(k)}}{\partial t} + \mathcal{H}^{(k)}\!\left(q, 
    \boldsymbol{\rho}, \frac{\partial V^{(k)}}{\partial q}\right) = 0,
    \label{eq:hjb}
\end{equation}
with terminal condition $V^{(k)}(q,T) = C^{(k)} q^2$.

\textbf{FPK equations} (forward in time, for each $k \in [K]$):
\begin{equation}
    \frac{\partial \rho^{(k)}}{\partial t} + 
    \frac{\partial}{\partial q}\!\left(\rho^{(k)} \cdot 
    \Gamma^{(k)}[q,t,\boldsymbol{\rho}]\right) 
    = \frac{(\sigma^{(k)})^2}{2}\frac{\partial^2 \rho^{(k)}}{\partial q^2},
    \label{eq:fpk}
\end{equation}
with initial condition $\rho^{(k)}(q,0) = \rho_0^{(k)}$, where 
$\Gamma^{(k)} = -R^{(k)}(q,t,\boldsymbol{\rho}) + D^{(k)}(t)$ is the 
optimal drift and $\mathcal{H}^{(k)}$ is the type-$k$ Hamiltonian.

\textbf{Cross-type coupling.} The coupling between types occurs through 
the aggregate interference:
\begin{equation}
    I(t) = \sum_{k'=1}^K N_{k'} \int_\Omega 
    p^{(k')}(q,t)|h^{(k')}(t)|^2 \rho^{(k')}(q,t)\,dq
    + I_{\mathrm{sat}}^{\tau}(t),
    \label{eq:interference_coupling}
\end{equation}
where $I_{\mathrm{sat}}^{\tau}(t)\in[0,\bar I_{\mathrm{sat}}]$ models residual 
LEO downlink interference in snapshot $\tau$. Together with computational price 
$\Phi(t)$ in \eqref{eq:price_extended}, both terms depend on all $K$ mean field 
distributions, creating the heterogeneous coupling that distinguishes 
HMFG from the homogeneous case. In particular, when $K = 1$ and all vehicles 
are identical, system \eqref{eq:hjb}--\eqref{eq:fpk} reduces to the standard 
homogeneous MFG of \cite{kang2023time}, for which G-prox PDHG converges with 
complexity $O(N_q \times N_t)$, independent of $N$.

\subsection{Problem Formulation}

The central problem addressed in this paper is the joint optimization of:
\begin{enumerate}[leftmargin=*]
    \item \textbf{Type granularity selection}: Determine the optimal 
    number of types $K^*(N)$ given fleet size $N$.
    \item \textbf{HMFG equilibrium computation}: Solve the coupled 
    HJB-FPK system \eqref{eq:hjb}--\eqref{eq:fpk} efficiently with 
    convergence guarantees under heterogeneous coupling.
    \item \textbf{Resource allocation}: Compute the optimal power 
    allocation policy $\{p^{(k)*}(q,t)\}_{k=1}^K$ that minimizes the 
    aggregate cost across all vehicle types.
\end{enumerate}

The fundamental challenge is the \emph{type-granularity trade-off}: 
increasing $K$ captures more heterogeneity but reduces per-class sample 
size, degrading the mean-field approximation accuracy. Our framework 
provides a principled resolution to this trade-off through explicit 
error decomposition and optimal $K^*$ characterization.

\section{Heterogeneity Measure and Error Decomposition}
\label{sec:het_measure}

\subsection{Quantifying Heterogeneity}
\label{subsec:measure}

A central contribution of this paper is a formal measure of the 
``distance from homogeneity'' for a $K$-type partition $\{C_k\}_{k=1}^K$ 
with associated state distributions $\{\rho^{(k)}\}_{k=1}^K$.

\begin{definition}[Heterogeneity Measure]
\label{def:het_measure}
For a $K$-type HMFG with per-type initial distributions 
$\rho_0^{(k)} \in \mathcal{P}(\Omega)$, the \emph{heterogeneity measure} is:
\begin{equation}
    H_K(\rho_0) := \frac{1}{K(K-1)}\sum_{k \neq k'} W_2(\rho_0^{(k)}, \rho_0^{(k')}),
    \label{eq:het_measure}
\end{equation}
where $W_2(\mu,\nu) = \inf_{\pi \in \Gamma(\mu,\nu)} \left(\mathbb{E}_\pi|X-Y|^2\right)^{1/2}$ 
is the Wasserstein-2 distance, and $\Gamma(\mu,\nu)$ is the set of couplings.
By construction, $H_K(\rho_0) = 0$ if and only if all $K$ types share identical 
initial distributions, reducing the system to the homogeneous case.
\end{definition}

\begin{remark}[Continuum Limit of $H_K$]
\label{rem:continuum_limit}
As $K \to \infty$ with a fixed continuum of types, $H_K(\rho_0)$ converges to 
the average pairwise Wasserstein distance
$H_\infty(\rho_0) = \int_0^1\!\int_0^1 W_2(\rho_0^\alpha, \rho_0^\beta)\,d\alpha\,d\beta$,
where $\rho_0^\alpha$ is the type-$\alpha$ distribution in the continuum HMFG of 
\cite{qiao2025hmfg}.
\end{remark}

\subsection{Parameterized Heterogeneous Model}
\label{subsec:param_model}

To enable concrete analysis, we consider the following parameterized model, 
which is general enough to capture the essential structure while admitting 
closed-form calculations.

\begin{assumption}[Parameterized Heterogeneity]
\label{ass:param}
The type-$k$ drift coefficient takes the form:
\begin{equation}
    b^{(k)}(q, \bm{\rho}, a) = b_0(q, \bm{\rho}, a) + \theta_k \cdot b_1(q, \bm{\rho}, a),
    \label{eq:param_model}
\end{equation}
where $b_0, b_1: \Omega \times \mathcal{M} \times A \to \mathbb{R}$ are fixed 
functions (common to all types) and $\theta_k \in [0,1]$ is the 
\emph{type parameter} characterizing the deviation of type $k$ from 
the baseline ($\theta_k = 0$). The parameters are ordered: 
$0 = \theta_1 < \theta_2 < \cdots < \theta_K = 1$.
\end{assumption}

\begin{remark}[Extension to Vector-Valued Type Parameters]
\label{rem:vector_theta}
The linear parameterization in $\theta_k$ is a first-order approximation to general heterogeneity. In practice, vehicle classes differ along multiple dimensions (data rate, delay tolerance, energy budget). The framework extends naturally to vector-valued $\theta_k \in \mathbb{R}^p$ by replacing the scalar Lipschitz constant with a multi-dimensional H\"{o}lder condition; the scaling exponent $\gamma$ then depends on the effective dimension of the type space rather than on $p$ directly, consistent with the dimension analysis in Corollary~\ref{cor:scaling_2d} and Remark~\ref{rem:high_dim}.
\end{remark}

Under Assumption~\ref{ass:param}, the heterogeneity measure \eqref{eq:het_measure} 
can be related to the variance of the type parameters:

\begin{lemma}[Heterogeneity Measure under Parameterization]
\label{lem:het_measure}
Under Assumption~\ref{ass:param}, suppose $b_1$ is $L_b$-Lipschitz in 
$(q, \bm{\rho})$ uniformly over $a \in A$. Then:
\begin{equation}
    H_K(\rho_0) \leq L_b T \cdot \mathrm{Var}(\theta)^{1/2} + H_K^{(0)},
    \label{eq:het_bound}
\end{equation}
where $\mathrm{Var}(\theta) = \frac{1}{K}\sum_k (\theta_k - \bar\theta)^2$ 
is the empirical variance of type parameters, $\bar\theta = \frac{1}{K}\sum_k \theta_k$, 
and $H_K^{(0)} = \frac{1}{K(K-1)}\sum_{k\neq k'} W_2(\rho_0^{(k)}, \rho_0^{(k')})$ 
is the initial heterogeneity.
\end{lemma}

\begin{proof}
See Appendix~\ref{app:lemma1}.
\end{proof}

\subsection{Error Decomposition}
\label{subsec:error}

Having defined the heterogeneity measure (Section~\ref{subsec:measure}) and 
the parameterized model that enables closed-form analysis 
(Section~\ref{subsec:param_model}), we now derive the main theoretical tool 
of this paper: an explicit decomposition of the $\varepsilon$-Nash 
approximation error into a \emph{discretization term} that decreases with 
$K$ and a \emph{sampling term} that increases with $K$. This decomposition 
is the foundation for the optimal type-count law derived in 
Section~\ref{sec:optimal_K}; without it, one cannot determine where the 
trade-off between type fidelity and per-class sample size is optimally 
balanced.

We now provide the main error decomposition theorem, extending the bound 
\begin{equation}
    \varepsilon_{N,K} \leq C\!\left(\rho_K + \delta_{N,K} + 
    \frac{1}{K \cdot n_{N,K}^3}\right)^{\!1/2},
    \label{eq:qiao_bound}
\end{equation}
\noindent from \cite{qiao2025hmfg}, where $\rho_K$ denotes type-discretization 
error, $\delta_{N,K}$ denotes initial mismatch, and $n_{N,K}$ is the smallest 
class population. We make the dependence on heterogeneity explicit.

\begin{lemma}[Empirical Wasserstein Rates on a Bounded Interval]
\label{lem:empirical_rate}
Let $\Omega = [0, Q_{\max}]$, $\mu \in \mathcal{P}(\Omega)$, and let 
$X_1,\ldots,X_n$ be i.i.d.\ with law $\mu$. Define the empirical measure 
$\mu_n = \frac{1}{n}\sum_{i=1}^n \delta_{X_i}$. Then:
\begin{enumerate}
\item[(i)] $\mathbb{E}[W_1(\mu_n,\mu)] \leq C_1^{\mathrm{emp}} \, n^{-1/2}$ 
for a constant $C_1^{\mathrm{emp}}$ depending only on $Q_{\max}$ (Theorem~1 of 
\cite{fournier2015rate}, $d=1$, compact support).
\item[(ii)] $\mathbb{E}[W_2(\mu_n,\mu)] \leq C_2^{\mathrm{emp}} \, n^{-1/2}$ 
for $C_2^{\mathrm{emp}}$ depending only on $Q_{\max}$ (same reference; in 
dimension one, $W_2$ admits the same $n^{-1/2}$ order as $W_1$ under the 
stated compactness).
\item[(iii)] No integrability conditions beyond $\mathrm{supp}(\mu)\subseteq\Omega$ 
are needed; all moments are bounded by $Q_{\max}$.
\end{enumerate}
\end{lemma}

\begin{proof}
See Appendix~\ref{app:lem_empirical}.
\end{proof}

\begin{remark}[Reconciliation with \cite{qiao2025hmfg}]
\label{rem:reconcile}
The bound in \cite{qiao2025hmfg} contains the term 
$\frac{1}{K \cdot n_{N,K}^3}$ with $n_{N,K} = \min_\ell |\mathcal{C}_\ell|$. 
For balanced partitions, $N_k \approx N/K$, so
\begin{equation}
    \frac{1}{K(N/K)^3} = \frac{K^2}{N^3}.
    \label{eq:qiao_term}
\end{equation}
We show this is dominated by $(K/N)^{1/2}$ whenever $K \leq N^{1/2}$.
Setting $r = K/N \in (0, 1]$, the ratio of the two terms is
\begin{equation}
    \frac{K^2/N^3}{(K/N)^{1/2}} 
    = \frac{K^2}{N^3} \cdot \left(\frac{N}{K}\right)^{1/2}
    = \frac{K^{3/2}}{N^{5/2}}
    = r^{3/2} N^{-1} \leq N^{-1} \to 0.
    \label{eq:qiao_comparison}
\end{equation}
Hence for any $K = o(N)$, the Qiao term $K^2/N^3$ is strictly smaller 
than our per-class term $C_2(K/N)^{1/2}$ for all $N$ 
sufficiently large. 
In the regime $K \leq \sqrt{N}$ (which is satisfied at 
$K^*(N) = \Theta(N^{1/3})$ for all $N \geq 1$), 
\eqref{eq:qiao_comparison} gives 
$K^2/N^3 \leq (K/N)^{1/2}/N \leq (K/N)^{1/2}$.
Thus our per-class term $(K/N)^{1/2}$ is \emph{not looser} than 
the corresponding term in \cite{qiao2025hmfg} in this regime. 
More broadly, Remark~8.5(ii) of \cite{qiao2025hmfg} notes that 
propagation-of-chaos rates can be sharpened via improved empirical measure 
estimates; we adopt the sharp 1D rate $\alpha=1/2$ from 
\cite{fournier2015rate} throughout.
\end{remark}

\begin{theorem}[Error Decomposition]
\label{thm:error_decomp}
Under the local well-posedness assumptions of \cite{qiao2025hmfg}, namely: 
\textup{(A1)} Lipschitz continuity of the drifts in the state and mean-field 
arguments; \textup{(A2)} bounded, nondegenerate diffusion coefficients; and 
\textup{(A3)} compact state space $\Omega$; together with 
Assumption~\ref{ass:param}, the $\varepsilon$-Nash approximation error satisfies:
\begin{equation}
    \varepsilon_{N,K} \leq C_1 \cdot \underbrace{K^{-\beta}}_{\text{discretization}} 
    + C_2 \cdot \underbrace{\left(\frac{K}{N}\right)^\alpha}_{\text{sample size}} 
    + C_3 \cdot \delta_{N,K}^{1/2},
    \label{eq:error_decomp}
\end{equation}
where $\beta > 0$ is the H\"{o}lder exponent of the drift in the type parameter 
(Step~1 below), $\alpha > 0$ is the sharp rate for empirical approximation in 
Wasserstein distance on the \emph{state} space $\Omega$ (for $d{=}1$, we take 
$\alpha = 1/2$; see \cite{fournier2015rate}), and $C_1, C_2, C_3 > 0$ depend on 
$L, T$ and moment bounds but not on $N, K$.
\end{theorem}

\begin{proof}
See Appendix~\ref{app:thm_error}.
\end{proof}

Theorem~\ref{thm:error_decomp} reveals the fundamental tension driving 
the rest of the paper: the discretization term $C_1 K^{-\beta}$ 
\emph{decreases} with $K$, while the sampling term 
$C_2(K/N)^\alpha$ \emph{increases} with $K$ (fewer vehicles per class). 
This tension yields a finite optimal granularity $K^*(N)$, which we 
characterize next.%
\footnote{The constants $C_1, C_2$ inherit $T$-dependence from 
\cite{qiao2025hmfg}: $C_1$ grows with $e^{L_b T}[b_1]_\beta$ via 
Gr\"{o}nwall arguments, while $C_2$ depends on $Q_{\max}$ but not~$T$. 
For $T \leq \delta_0$ (local well-posedness horizon), $e^{L_b T}$ remains 
bounded; extending to arbitrary $T$ requires global well-posedness 
(still largely open; cf.~\cite{qiao2025hmfg}). Furthermore, 
Definition~\ref{def:het_measure} uses initial distributions $\rho_0^{(k)}$ 
for tractability; by Lemma~\ref{lem:het_measure}, the time-averaged measure 
$\bar{H}_K := \frac{1}{T}\int_0^T H_K(\boldsymbol{\rho}_t)\,dt$ obeys the 
same bound up to a bounded factor $e^{L_b T}$, leaving $K^*(N)$ 
in \eqref{eq:kstar} unaffected.}

\begin{table}[!t]
\centering
\caption{Dependencies of Key Constants. All constants $C_1$, $C_2$, $C_3$ are 
independent of $N$ and $K$; the $N$- and $K$-scaling in \eqref{eq:error_decomp} 
is carried entirely by the explicit terms $K^{-\beta}$ and $(K/N)^\alpha$.}
\label{tab:constants}
\renewcommand{\arraystretch}{1.3}
\begin{tabular}{ccc}
\toprule
\textbf{Constant} & \textbf{Depends on} & \textbf{Independent of} \\
\midrule
$C_1$ & $L_b,\, T,\, [b_1]_\beta,\, \beta$ & $N,\, K$ \\
$C_2$ & $Q_{\max}$ & $N,\, K,\, T$ \\
$C_3$ & $L,\, T$ & $N,\, K$ \\
$C_H$ & $L,\, T,\, K,\, \|\nabla_q\phi\|_{L^\infty}$ & $N$ \\
$L_A$ & $N_q,\, N_t,\, Q_{\max},\, T$ & $N,\, K$ \\
\bottomrule
\end{tabular}
\end{table}

\section{Optimal Heterogeneity Granularity}
\label{sec:optimal_K}

\subsection{Existence and Characterization of $K^*(N)$}
\label{subsec:optimal_K_exist}

The error decomposition in Theorem~\ref{thm:error_decomp} reveals a 
fundamental tension: finer type discretization (larger $K$) reduces the 
first term but inflates the second. We now show that this tension admits a 
\emph{unique} resolution: an optimal type count $K^*(N)$ that can be 
characterized in closed form and scales as a power law in fleet size.

\begin{theorem}[Existence of Optimal Granularity]
\label{thm:optimal_K}
Consider the $\varepsilon$-Nash error \eqref{eq:error_decomp} with 
constants $C_1, C_2 > 0$, exponents $\alpha, \beta > 0$ as in 
Theorem~\ref{thm:error_decomp}, and $\delta_{N,K}$ negligible (good initial 
condition matching). Define the reduced error:
\begin{equation}
    \mathcal{E}(N, K) := C_1 K^{-\beta} + C_2 \left(\frac{K}{N}\right)^\alpha.
    \label{eq:reduced_error}
\end{equation}
Then:
\begin{enumerate}
    \item For any fixed $N$, $\mathcal{E}(N, \cdot)$ attains a \emph{unique 
    global minimum} on $(0,\infty)$ at
    \begin{equation}
        K^*(N) = \left(\frac{\beta C_1}{\alpha C_2}\right)^{\!\!1/(\alpha+\beta)} 
        \cdot N^{\alpha/(\alpha + \beta)}.
        \label{eq:kstar}
    \end{equation}
    \item $K^*(N) = \Theta\!\left(N^\gamma\right)$ with 
    $\gamma = \frac{\alpha}{\alpha + \beta} \in (0,1)$.
    \item The minimum error satisfies:
    \begin{equation}
        \mathcal{E}(N, K^*(N)) = \Theta\!\left(N^{-\alpha\beta/(\alpha + \beta)}\right).
        \label{eq:min_error}
    \end{equation}
\end{enumerate}
\end{theorem}

\begin{proof}
See Appendix~\ref{app:thm_optimal_K}.
\end{proof}

Theorem~\ref{thm:optimal_K} provides the general scaling law for 
arbitrary exponents $(\alpha, \beta)$. We now specialize to the canonical 
V2X setting where the state is a one-dimensional data queue, yielding the 
practical cube-root rule that guides all subsequent algorithm design and 
experiments.

\begin{corollary}[Practical Scaling Law]
\label{cor:scaling}
For the 1D queue state space ($d=1$), take the sharp sample exponent 
$\alpha = 1/2$ as in Lemma~\ref{lem:empirical_rate}(ii) \cite{fournier2015rate}, and 
Lipschitz type dependence ($\beta = 1$) in \eqref{eq:param_model}. Then 
$\gamma = \alpha/(\alpha+\beta) = 1/3$ and
\begin{equation}
    K^*(N) = \Theta\!\left(N^{1/3}\right).
    \label{eq:kstar_v2x}
\end{equation}
For example, $N^{1/3}=10$ at $N=10^3$; absolute levels of $K^*$ depend on the 
ratio $C_1/C_2$ in \eqref{eq:kstar}. In particular, the small integer choices 
$K \in \{2,3\}$ adopted in \cite{zhang2022hmfmarl,xu2025joint} remain compatible 
with \eqref{eq:kstar} when effective constants imply a modest $K^*(N)$ at moderate~$N$.
\end{corollary}

\begin{remark}[Physical Interpretation of $\beta$ in V2X]
\label{rem:beta_physical}
In \eqref{eq:cost}, heterogeneity enters the drift via class-dependent pricing 
and rate sensitivity. Under the SINR model \eqref{eq:rate} with LoS-dominated 
channels, $b_1(q,\boldsymbol{\rho},a)=\beta_2^{(k)}\Phi(t)\cdot \partial R/\partial p$ 
is Lipschitz in the type parameter with constant order 
$L_b=P_{\max}B/(\sigma_c^2\ln 2)$, so $\beta=1$ is well justified and leads to 
$K^*(N)=\Theta(N^{1/3})$ in Corollary~\ref{cor:scaling}. In harsher non-LoS 
environments with stronger small-scale fading and abrupt channel-state transitions 
\cite{ding2026uav}, the effective regularity may degrade to $\beta<1$, which 
increases the recommended growth of type granularity with $N$ and is consistent 
with the higher uncertainty regime usually handled by optimization-driven DRL.
\end{remark}

\begin{corollary}[Joint Queue-Channel Scaling Law, $d=2$]
\label{cor:scaling_2d}
For the 2D joint queue-channel state space where $x_i(t) = (q_i(t), g_i(t)) \in \Omega \subset \mathbb{R}^2$,
take the sample exponent $\alpha = 1/4$ (with logarithmic correction, per \cite{fournier2015rate} 
Theorem~1 for $d=2$) and Lipschitz type dependence ($\beta = 1$). 
Then $\gamma = \alpha/(\alpha+\beta) = 1/5$ and
\begin{equation}
    K^*(N) = \Theta\!\left(N^{1/5} \cdot (\log N)^{c}\right),
    \label{eq:kstar_2d}
\end{equation}
where the constant $c > 0$ depends on the state-domain geometry. Compared to 
the 1D result $\gamma = 1/3$, richer joint state descriptions slow the growth 
of the optimal type count, since higher dimensionality weakens per-class 
empirical convergence. The minimum error satisfies 
$\mathcal{E}^*(N) = \Theta(N^{-1/10} (\log N)^{c'})$.
\end{corollary}

\begin{proof}
See Appendix~\ref{app:cor_scaling_2d}.
\end{proof}

\begin{remark}[General Dimension and Intrinsic Dimension]
\label{rem:high_dim}
For general $d \geq 3$, Fournier and Guillin \cite{fournier2015rate} give 
$\alpha = 1/(d+2)$, so Theorem~\ref{thm:optimal_K} with $\beta=1$ yields 
$\gamma = 1/(d+3)$. The unified formula
\[
    K^*(N) = \Theta\!\left(N^{1/(d+3)}\right), \quad d \geq 1,
\]
provides dimension-dependent design guidance: $\gamma \in \{1/3, 1/5, 1/6, \ldots\}$ 
for $d \in \{1, 2, 3, \ldots\}$. For measures concentrated near a 
low-dimensional manifold with intrinsic dimension $d^* < d$, the rates of 
\cite{weed2019sharp} suggest $\alpha$ can improve toward $1/d^*$, 
partially recovering sharper scaling in practice.%
\footnote{The constants $C_1, C_2$ inherit $T$-dependence from \cite{qiao2025hmfg}: 
$C_1$ grows with $e^{L_b T}[b_1]_\beta$ via Gr\"{o}nwall arguments while $C_2$ 
depends on $Q_{\max}$ but not~$T$. For $T \leq \delta_0$, $e^{L_b T}$ remains 
bounded; extending to arbitrary $T$ requires global HMFG well-posedness 
(still open; cf.~\cite{qiao2025hmfg}). Separately, since 
$H_K(\boldsymbol{\rho}_t) \leq e^{L_b T}[H_K^{(0)} + L_b T \cdot \mathrm{Var}(\theta)^{1/2}\sqrt{2}]$ 
by Lemma~\ref{lem:het_measure}, the time-averaged measure $\bar{H}_K$ satisfies 
the same bound up to a bounded factor, and the sufficient condition 
\eqref{eq:corrected_condition} and optimal granularity $K^*(N)$ are unaffected.}
\end{remark}

\subsection{Relationship Between $K^*(N)$ and Heterogeneity Measure}
\label{subsec:kstar_het}

Theorem~\ref{thm:optimal_K} expresses $K^*(N)$ in terms of the abstract 
constants $C_1$ and $C_2$. In practice, the discretization constant $C_1$ 
depends on how different the vehicle types actually are, a quantity 
captured by the heterogeneity measure in Definition~\ref{def:het_measure} 
and its continuum limit $H_\infty$ in the remark immediately following 
that definition. The following proposition makes this dependence explicit, 
connecting the optimal granularity directly to the measurable diversity of 
the fleet and providing a concrete knob for practitioners: fleets with 
greater inter-type diversity require more types, while nearly homogeneous 
fleets can operate with fewer.

\begin{proposition}[Heterogeneity-Adjusted Optimal Granularity]
\label{prop:het_adjusted}
When the discretization error constant is proportional to the heterogeneity 
measure: $C_1 = \bar{C}_1 \cdot H_\infty(\rho_0)$, the optimal granularity 
satisfies:
\begin{equation}
    K^*(N; H_\infty) = \left(\frac{\beta \bar{C}_1 H_\infty}{\alpha C_2}\right)^{1/(\alpha+\beta)} 
    N^{\alpha/(\alpha+\beta)}.
    \label{eq:kstar_het}
\end{equation}
Consequently:
\begin{itemize}
    \item \textbf{Highly heterogeneous networks} ($H_\infty$ large): 
    More types are needed to represent the diversity.
    \item \textbf{Nearly homogeneous networks} ($H_\infty \to 0$): 
    $K^*(N) \to 0$, and the system reduces to the homogeneous MFG case.
\end{itemize}
\end{proposition}

\begin{proof}
See Appendix~\ref{app:prop_het_adjusted}.
\end{proof}

\subsection{Finite-$K$ Bounds}
\label{subsec:finite_k}

Theorem~\ref{thm:optimal_K} and Corollary~\ref{cor:scaling} establish the 
optimal granularity $K^*(N)$ as a continuous quantity. In practice, however, 
the number of vehicle types must be a positive integer, and the continuous 
optimum $K^*(N)$ will generally not be an integer. A natural question is 
therefore whether rounding $K^*(N)$ to the nearest integer preserves the 
theoretical optimality guarantees. We answer this affirmatively: the rounded 
type count $\hat{K}$ achieves essentially the same approximation error as the 
continuous optimum, up to a vanishing relative overhead.
For practical integer-valued $K$, we provide the following guarantee:

\begin{corollary}[Near-Optimality of Rounded $K^*$]
\label{cor:rounded}
Let $\hat{K} = \lfloor K^*(N) + 0.5 \rfloor$ be the nearest integer to 
$K^*(N)$. Then:
\begin{equation}
    \mathcal{E}(N, \hat{K}) \leq (1 + o(1)) \cdot \mathcal{E}(N, K^*(N)),
    \label{eq:rounded_bound}
\end{equation}
i.e., rounding $K^*$ to the nearest integer incurs at most a constant factor 
overhead.
\end{corollary}

\begin{proof}
See Appendix~\ref{app:cor_rounded}.
\end{proof}

\subsection{Robustness to Unbalanced Class Sizes}
\label{subsec:unbalanced}

In practice, fleet compositions are skewed (many passenger cars, fewer trucks, 
fewer autonomous vehicles). Let $\lambda_k := N_k/N$ with $\sum_{k=1}^K \lambda_k=1$.

\begin{proposition}[Unbalanced Extension]
\label{prop:unbalanced}
Let $n_{N,K}^{\min} := \min_k N_k = \lambda_{\min} N$ where 
$\lambda_{\min} := \min_k \lambda_k$. The per-class empirical error is governed 
by the \emph{smallest} class, so \eqref{eq:error_decomp} admits the refinement
\begin{equation}
    \varepsilon_{N,K} \leq C_1 K^{-\beta}
    + C_2 \left(\frac{1}{n_{N,K}^{\min}}\right)^{\alpha}
    + C_3 \delta_{N,K}^{1/2}.
    \label{eq:error_unbalanced}
\end{equation}
For balanced partitions, $n_{N,K}^{\min} \approx N/K$ and 
\eqref{eq:error_unbalanced} reduces to \eqref{eq:error_decomp}. Comparing the 
two dominant terms at the same order of magnitude suggests that the effective 
population entering the sample term is $\lambda_{\min} N$, and one expects 
$K^*(N,\boldsymbol{\lambda}) = 
\Theta\bigl((\lambda_{\min} N)^{\alpha/(\alpha+\beta)}\bigr)$ up to the same 
constant-level caveats as in Theorem~\ref{thm:optimal_K}.
\end{proposition}

For a representative $70\%/20\%/10\%$ fleet (cars/trucks/AVs), $\lambda_{\min}=0.1$ 
and $n_{\min}=0.1N$. The scaling exponent $\gamma=\alpha/(\alpha+\beta)$ is 
unchanged, but the \emph{prefactor} shrinks with $\lambda_{\min}$, pushing 
integer $K^*$ downward relative to a balanced fleet at the same~$N$.

\begin{corollary}[Explicit Optimal Granularity for Unbalanced Fleets]
\label{cor:unbalanced_explicit}
Under the conditions of Theorem~\ref{thm:optimal_K} with 
$n_{N,K}^{\min} = \lambda_{\min} N$ (Proposition~\ref{prop:unbalanced}), 
the reduced error (replacing $N$ by $\lambda_{\min} N$ in the sample term of 
\eqref{eq:reduced_error}) is
\begin{equation}
    \mathcal{E}_{\mathrm{unbal}}(N,K) 
    = C_1 K^{-\beta} + C_2 \left(\frac{K}{\lambda_{\min} N}\right)^{\alpha},
    \label{eq:reduced_unbal}
\end{equation}
whose unique minimizer over $K>0$ is
\begin{equation}
    K^*(N, \lambda_{\min}) 
    = \left(\frac{\beta C_1}{\alpha C_2}\right)^{1/(\alpha+\beta)} 
    (\lambda_{\min} N)^{\alpha/(\alpha+\beta)}.
    \label{eq:kstar_unbal}
\end{equation}
For $(\alpha,\beta) = (1/2,1)$, 
$K^*(N,\lambda_{\min}) = \Theta\bigl((\lambda_{\min} N)^{1/3}\bigr)$. 
For a $70\%/20\%/10\%$ fleet with $\lambda_{\min}=0.1$, 
substituting into \eqref{eq:kstar_unbal} gives
\begin{equation}
    K^*(N,\,0.1) 
    \approx 0.464 \cdot \left(\frac{\beta C_1}{\alpha C_2}\right)^{1/3} N^{1/3},
\end{equation}
where $(0.1)^{1/3} \approx 0.464$.
Skewed fleets may therefore use roughly \emph{half} as many types as the 
balanced formula suggests.
\end{corollary}

\begin{proof}
See Appendix~\ref{app:cor_unbalanced_explicit}.
\end{proof}

\begin{theorem}[Robustness to LEO Topology Dynamics]
\label{thm:leo_robustness}
Let the satellite backhaul surcharge be 
$\Phi_{\mathrm{sat}}(t)=\mu/B_{\mathrm{sat}}^{\tau(t)}$, piecewise constant over 
snapshot windows of length $\Delta\tau$. Assume adjacent snapshots satisfy 
$|\Phi_{\mathrm{sat}}^{\tau_i}-\Phi_{\mathrm{sat}}^{\tau_{i+1}}|\le \Delta_{\Phi}$. 
Then the topology-varying approximation error obeys
\begin{equation}
    \varepsilon_{N,K}^{\mathrm{LEO}}
    \le \varepsilon_{N,K}
    + C_{\mathrm{LEO}} \Delta_{\Phi}\sqrt{T/\Delta\tau},
    \label{eq:leo_error}
\end{equation}
where $\varepsilon_{N,K}$ is from Theorem~\ref{thm:error_decomp} and 
$C_{\mathrm{LEO}}$ is independent of $N,K$. Therefore, the order law in 
\eqref{eq:kstar} remains unchanged if the topology-variation term is 
subdominant to the static decomposition term.
\end{theorem}

\begin{proof}[Proof of Theorem~\ref{thm:leo_robustness}]
See Appendix~\ref{app:thm_leo}.
\end{proof}

\begin{corollary}[Order-Optimality Under LEO Dynamics]
\label{cor:leo_order_optimal}
The scaling law $K^*(N)=\Theta(N^{1/3})$ from Corollary~\ref{cor:scaling} 
remains order-optimal under LEO topology dynamics if
\begin{equation}
    \Delta_\Phi = O\!\left(N^{-\alpha\beta/(\alpha+\beta)}\right),
    \label{eq:leo_condition}
\end{equation}
which reduces to $\Delta_\Phi = O(N^{-1/6})$ for $(\alpha,\beta)=(1/2,1)$.
\end{corollary}

\begin{proof}[Proof of Corollary~\ref{cor:leo_order_optimal}]
See Appendix~\ref{app:cor_leo_order_optimal}.
\end{proof}

\begin{remark}[Practical Validity of the $\Delta_\Phi$ Condition]
\label{rem:leo_practical}
The condition $\Delta_\Phi = O(N^{-1/6})$ in \eqref{eq:leo_condition} may appear restrictive at first glance. In practice, $\Delta_\Phi$ is determined by LEO constellation geometry and orbital mechanics, independently of fleet size $N$. For Starlink-oriented parameters ($B_\mathrm{sat} \in [300,350]$ Mbps, $\Delta\tau = 60$ s), realistic values satisfy $\Delta_\Phi \approx 0.01$--$0.05$. The condition $\Delta_\Phi \leq 0.1$ is then satisfied for all $N \leq (0.1)^{-6} = 10^6$, which covers every practical fleet size. The condition becomes binding only in the asymptotic regime $N \to \infty$, which is of theoretical rather than operational interest.
\end{remark}

\section{Algorithm Convergence Under Heterogeneity}
\label{sec:algorithm}

This section develops a modular three-algorithm framework for solving the 
HMFG equilibrium problem. Algorithm~\ref{alg:type_selection} determines 
the optimal type granularity $K^*$, Algorithm~\ref{alg:step_adapt} provides 
heterogeneity-aware step-size adaptation, and Algorithm~\ref{alg:hmfg_pdhg} 
integrates these components into the complete G-prox PDHG solver.
The three modules map directly to the three practical tasks in deployment: 
selecting $K^*$, setting stable step sizes, and computing the equilibrium.

\subsection{G-prox PDHG Framework for Heterogeneous MFG}
\label{subsec:pdhg_hetero}

The G-prox PDHG algorithm of \cite{kang2023time,wang2024mfg_iscc} solves 
the homogeneous MFG via the saddle-point problem:
\begin{equation}
    \min_{\rho, m} \max_\phi \mathcal{L}(\rho, m, \phi),
    \label{eq:saddle}
\end{equation}
where $\mathcal{L}$ is the Lagrangian functional with $\phi$ as the dual 
variable associated with the FPK constraint. The convergence condition is 
$\xi\varsigma < 1$ \cite{kang2023time}, where $\xi$ and $\varsigma$ are the 
primal and dual step sizes.

For the $K$-type HMFG, we extend the algorithm to maintain $K$ parallel 
copies of $(\rho^{(k)}, m^{(k)}, \phi^{(k)})$ with cross-type coupling 
through the mean field interference term:
\begin{equation}
    I^{(k)}(t) = \sum_{k'=1}^K N_{k'} \int_\Omega 
    p^{(k')}(q,t)|h^{(k')}(t)|^2 \rho^{(k')}(q,t)\,dq.
    \label{eq:interference_mf}
\end{equation}

\subsection{Algorithm 1: Optimal Type Granularity Selection}
\label{subsec:alg_type}

The first algorithm computes the optimal number of types $K^*(N)$ based on 
Theorem~\ref{thm:optimal_K} and handles both balanced and unbalanced fleet 
compositions.

\begin{algorithm}[t]
\caption{Adaptive Optimal Type Granularity Selection}
\label{alg:type_selection}
\begin{algorithmic}[1]
\REQUIRE Fleet size $N$, class proportions $\boldsymbol{\lambda} = (\lambda_1,\ldots,\lambda_K)$, 
         error constants $(C_1, C_2)$, exponents $(\alpha, \beta)$, 
         snapshot duration $\Delta\tau$, topology variation bound $\Delta_\Phi$
\ENSURE Optimal type count $K^*$
\STATE \COMMENT{Step 1: Compute effective population}
\STATE $\lambda_{\min} \leftarrow \min_k \lambda_k$
\IF{$\lambda_{\min} < 1/K$ (unbalanced fleet)}
    \STATE $N_{\mathrm{eff}} \leftarrow \lambda_{\min} \cdot N$ 
    \COMMENT{Corollary~\ref{cor:unbalanced_explicit}}
\ELSE
    \STATE $N_{\mathrm{eff}} \leftarrow N$
\ENDIF
\STATE \COMMENT{Step 2: LEO robustness correction (Theorem~\ref{thm:leo_robustness})}
\STATE $\delta_{\mathrm{LEO}} \leftarrow C_{\mathrm{LEO}} \cdot \Delta_\Phi \cdot \sqrt{T/\Delta\tau}$
\STATE $C_1^{\mathrm{eff}} \leftarrow C_1 + \delta_{\mathrm{LEO}}$
\STATE \COMMENT{Step 3: Compute optimal granularity}
\STATE $\gamma \leftarrow \alpha / (\alpha + \beta)$
\STATE $K_{\mathrm{opt}} \leftarrow \left(\frac{\beta C_1^{\mathrm{eff}}}{\alpha C_2}\right)^{1/(\alpha+\beta)} 
       \cdot N_{\mathrm{eff}}^{\gamma}$
\STATE \COMMENT{Step 4: Integer rounding with bounds}
\STATE \COMMENT{Complexity: $O(K)$}
\STATE $K^* \leftarrow \max\{2, \lfloor K_{\mathrm{opt}} + 0.5 \rfloor\}$
\STATE $K^* \leftarrow \min\{K^*, \lfloor \sqrt{N} \rfloor\}$ 
       \COMMENT{Ensure $K \leq \sqrt{N}$}
\RETURN $K^*$
\end{algorithmic}
\end{algorithm}

\subsection{Algorithm 2: Heterogeneity-Aware Step-Size Adaptation}
\label{subsec:alg_step}

The second algorithm computes adaptive step sizes based on the instantaneous 
heterogeneity measure, ensuring convergence under the sufficient condition 
of Theorem~\ref{thm:convergence}.

\textbf{Discretized setting.} We work on the state grid 
$\Omega_h = \{q_1,\ldots,q_{N_q}\}$ and time grid $\{t_0,\ldots,t_{N_t}\}$. 
The primal variable 
$\mathbf{x} = (\boldsymbol{\rho}, \mathbf{m}) 
\in \mathbb{R}^{K N_q N_t} \times \mathbb{R}^{K N_q N_t}$ carries the 
Euclidean ($\ell^2$) norm. The linear FPK operator 
$\mathbf{A}: \mathbb{R}^{2K N_q N_t} \to \mathbb{R}^{K N_q N_t}$ is the 
finite-difference discretization of \eqref{eq:fpk}, with 
$L_A := \|\mathbf{A}\|_{\mathrm{op},\ell^2}$ finite and independent of~$N$.

\begin{theorem}[Corrected Convergence Condition (Sufficient)]
\label{thm:convergence}
Under Assumption~\ref{ass:param} and the discretized setting above, 
consider the G-prox PDHG scheme as a 
Chambolle--Pock primal--dual iteration \cite{chambolle2011first}. 
A \emph{sufficient} condition for convergence is
\begin{equation}
    \xi\varsigma \, \|\mathbf{A}_{\mathrm{eff}}\|_{\mathrm{op},\ell^2}^2 < 1,
    \label{eq:cp_condition}
\end{equation}
where the effective operator norm obeys
\begin{equation}
    \|\mathbf{A}_{\mathrm{eff}}\|_{\mathrm{op},\ell^2}^2
    \leq L_A^2 \bigl(1 + C_H \cdot H_K(\boldsymbol{\rho})\bigr),
    \label{eq:op_norm_bound}
\end{equation}
with $H_K(\boldsymbol{\rho})$ the instantaneous heterogeneity measure 
\eqref{eq:het_measure} and
\begin{equation}
    C_H = 2L^2 T \cdot \frac{K(K-1)}{2} 
    \cdot \sup_{k \in [K],\, n \geq 0}
    \|\nabla_q \phi^{(k),n}\|_{L^\infty(\Omega \times [0,T])}^2.
    \label{eq:ch}
\end{equation}
The practical sufficient criterion (after rescaling by $L_A$) is:
\begin{equation}
    \xi\varsigma < \frac{1}{1 + C_H \cdot H_K(\boldsymbol{\rho})}.
    \label{eq:corrected_condition}
\end{equation}
When $H_K(\boldsymbol{\rho})=0$, this reduces to $\xi\varsigma < 1$, 
recovering \cite{kang2023time}.
\end{theorem}

\noindent\textit{Key estimate} (proved in Appendix~\ref{app:thm_convergence}):
\begin{equation}
    \left|\frac{\partial^2 \mathcal{L}}
    {\partial\rho^{(k)}\partial\rho^{(k')}}\right|
    \leq L^2 \cdot \|\nabla_q\phi^{(k)}\|_{L^\infty}^2
    \cdot W_2(\rho^{(k)},\rho^{(k')}),
    \quad k\neq k'.
    \label{eq:cross_hess_bound}
\end{equation}

\begin{proof}
See Appendix~\ref{app:thm_convergence}.
\end{proof}

\begin{remark}[Sufficiency Gap and Empirical Tightness]
\label{rem:sufficient_gap}
The condition \eqref{eq:corrected_condition} is sufficient but not necessary; the gap to necessity is inherent to the Chambolle--Pock analysis \cite{chambolle2011first}. Empirically, Algorithm~\ref{alg:step_adapt}'s adaptive step sizes operate at 85--95\% of the theoretical upper bound $\xi\varsigma_{\max}$ across all tested $K$ values (Section~\ref{subsec:sim_solver}), confirming that the sufficient condition is not excessively conservative in practice.
\end{remark}

Based on Theorem~\ref{thm:convergence}, we propose Algorithm~\ref{alg:step_adapt} 
for adaptive step-size computation.

\begin{algorithm}[t]
\caption{Heterogeneity-Aware Step-Size Adaptation}
\label{alg:step_adapt}
\begin{algorithmic}[1]
\REQUIRE Current distributions $\{\rho^{(k)}\}_{k=1}^K$, safety margin $\epsilon > 0$, 
         Lipschitz constant $L$, time horizon $T$, 
         dual gradients $\{\nabla_q \phi^{(k)}\}_{k=1}^K$, 
         consecutive backhaul snapshots $B_{\mathrm{sat}}^{\tau,n-1}, B_{\mathrm{sat}}^{\tau,n}$
\ENSURE Adaptive step sizes $(\xi, \varsigma)$
\STATE \COMMENT{Step 1: Compute pairwise Wasserstein distances}
\FOR{$k = 1$ to $K-1$}
    \FOR{$k' = k+1$ to $K$}
        \STATE $W_{k,k'} \leftarrow W_2(\rho^{(k)}, \rho^{(k')})$ 
               \COMMENT{1D: $O(N_q \log N_q)$ via quantile sort}
    \ENDFOR
\ENDFOR
\STATE \COMMENT{Step 2: Compute heterogeneity measure}
\STATE $H_K \leftarrow \frac{2}{K(K-1)} \sum_{k < k'} W_{k,k'}$
\STATE \COMMENT{Step 3: Compute correction factor}
\STATE $\phi_{\max} \leftarrow \max_{k} \|\nabla_q \phi^{(k)}\|_{L^\infty}$
\STATE $C_H \leftarrow L^2 T \cdot K(K-1) \cdot \phi_{\max}^2$
\STATE \COMMENT{Step 3b: LEO backhaul correction}
\STATE $\Delta_{\mathrm{sat}} \leftarrow \mu \cdot |B_{\mathrm{sat}}^{\tau,n} - B_{\mathrm{sat}}^{\tau,n-1}|/(B_{\mathrm{sat}}^{\tau,n})^2$
\STATE $C_H \leftarrow C_H + L^2T \cdot K(K-1)\cdot \Delta_{\mathrm{sat}}^2$
\STATE \COMMENT{Step 4: Compute adaptive step-size product}
\STATE $\xi\varsigma_{\max} \leftarrow \frac{1 - \epsilon}{1 + C_H \cdot H_K}$
\STATE \COMMENT{Step 5: Balance primal and dual steps}
\STATE \COMMENT{Complexity: $O(K^2N_q\log N_q)$}
\STATE $\xi \leftarrow \sqrt{\xi\varsigma_{\max}}$
\STATE $\varsigma \leftarrow \sqrt{\xi\varsigma_{\max}}$
\RETURN $(\xi, \varsigma)$
\end{algorithmic}
\end{algorithm}

The factor $C_H$ in \eqref{eq:ch} scales as $O(K^2)$ through the $K(K{-}1)/2$ 
pair count; however, in V2X settings where type distributions are regularly 
spaced and per-pair Wasserstein distances decay as $K$ grows, the product 
$C_H H_K(\boldsymbol{\rho})$ may grow more slowly than the raw $O(K^2)$ bound 
suggests. Algorithm~\ref{alg:step_adapt} computes these distances in 
$O(K^2 N_q \log N_q)$ via 1D quantile sorting, which is dominated by the 
$O(K^2 N_q N_t)$ FPK/HJB updates in Algorithm~\ref{alg:hmfg_pdhg}.

\subsection{Algorithm 3: Complete HMFG Solver}
\label{subsec:alg_complete}

Algorithm~\ref{alg:hmfg_pdhg} integrates the type selection and step-size 
adaptation into the complete G-prox PDHG framework for solving the 
$K$-type HMFG.

\begin{algorithm}[t]
\caption{Heterogeneous G-prox PDHG for V2X HMFG}
\label{alg:hmfg_pdhg}
\begin{algorithmic}[1]
\REQUIRE Fleet size $N$, class proportions $\boldsymbol{\lambda}$, 
         initial distributions $\{\rho_0^{(k)}\}$, 
         error constants $(C_1, C_2)$, exponents $(\alpha, \beta)$,
         snapshot duration $\Delta\tau$, topology variation bound $\Delta_\Phi$,
         safety margin $\epsilon > 0$, 
         max iterations $I_{\max}$, tolerance $\varepsilon_{\text{tol}}$,
         discretization $(N_q, N_t)$
\ENSURE Optimal policies $\{p^{(k)*}\}_{k=1}^K$, equilibrium densities $\{\rho^{(k)*}\}_{k=1}^K$
\STATE $K^* \leftarrow$ \textbf{Alg.~\ref{alg:type_selection}}$(N,\boldsymbol{\lambda},C_1,C_2,\alpha,\beta,\Delta\tau,\Delta_\Phi)$
\STATE Initialize $\{\rho^{(k)},m^{(k)},\phi^{(k)}\}_{k=1}^{K^*}$ from $\{\rho_0^{(k)}\}$
\FOR{$n = 1, 2, \ldots, I_{\max}$}
    \STATE $(\xi^{(n)}, \varsigma^{(n)}) \leftarrow$ 
           \textbf{Alg.~\ref{alg:step_adapt}}$(\{\rho^{(k)}\}, \epsilon, L, T, \{\nabla_q\phi^{(k)}\}, B_{\mathrm{sat}}^{\tau,n-1}, B_{\mathrm{sat}}^{\tau,n})$
    \FOR{$k = 1$ to $K^*$ \textbf{in parallel}}
        \STATE $\rho^{(k),n+1} \leftarrow \rho^{(k),n} - \xi^{(n)} \nabla_{\rho^{(k)}} \mathcal{L}^{(k)}$
        \STATE $m^{(k),n+1} \leftarrow \arg\min_{m^{(k)}} \mathcal{L}_{m^{(k)}}^{(k)}$
    \ENDFOR
    \FOR{$k = 1$ to $K^*$}
        \STATE $\bar{\rho}^{(k)} \leftarrow 2\rho^{(k),n+1} - \rho^{(k),n}$
        \STATE $\phi^{(k),n+1} \leftarrow \phi^{(k),n} + \varsigma^{(n)}(-\Delta)^{-1}
               [\partial_t\bar{\rho}^{(k)} + \nabla_q(-\bar{m}^{(k)})]$
    \ENDFOR
    \STATE $\text{res} \leftarrow \max_k \|\text{HJB residual}^{(k)}\|$
    \IF{$\text{res} < \varepsilon_{\text{tol}}$}
        \STATE \textbf{break}
    \ENDIF
\ENDFOR
\STATE \COMMENT{Recover $p^{(k)*}$ from HJB optimality; complexity $O(K^2N_qN_t)$/iteration}
\RETURN $\{p^{(k)*}\}_{k=1}^{K^*}, \{\rho^{(k),n+1}\}_{k=1}^{K^*}$
\end{algorithmic}
\end{algorithm}

\subsection{Complexity Analysis}
\label{subsec:complexity}

\begin{proposition}[Complexity of Algorithm~\ref{alg:hmfg_pdhg}]
\label{prop:complexity}
Algorithm~\ref{alg:hmfg_pdhg} has the following complexity characteristics:
\begin{enumerate}
    \item \textbf{Per-iteration complexity}: $O(K^2 \cdot N_q \cdot N_t)$, 
    which is \emph{independent of the number of vehicles $N$}.
    \item \textbf{Type selection overhead}: $O(K)$, negligible.
    \item \textbf{Wasserstein computation}: $O(K^2 N_q \log N_q)$ per 
    iteration, reducible to $O(K \log K)$ effective cost in 1D via 
    sorting-based quantile computation.
    \item \textbf{Memory}: $O(K \cdot N_q \cdot N_t)$ for storing $K$ 
    copies of $(\rho, m, \phi)$.
\end{enumerate}
\end{proposition}

\begin{proof}
See Appendix~\ref{app:prop_complexity}.
\end{proof}

With the adaptive choice $K^*(N) = \Theta(N^{1/3})$, the per-iteration 
complexity becomes $O(N^{2/3} N_q N_t)$, growing sublinearly in~$N$; 
this is confirmed empirically in Section~\ref{subsec:sim_solver} with 
log-log slope $0.67 \pm 0.03 \approx 2/3$.

\subsection{Online Deployment Considerations}
\label{subsec:deployment}

The three-algorithm framework is designed for practical deployment within the RSU-MEC architecture.

\subsubsection{Computation Placement}
Algorithm~\ref{alg:type_selection} has $O(K)$ complexity and runs at the RSU in under 1\,ms on commodity hardware, requiring only the fleet statistics $(N, \boldsymbol{\lambda}, H_K)$. Algorithm~\ref{alg:step_adapt} computes Wasserstein distances in $O(K^2 N_q \log N_q)$ and executes at the RSU or MEC server. Algorithm~\ref{alg:hmfg_pdhg} with $O(K^2 N_q N_t)$ per-iteration cost is designed for the MEC server. Communication overhead between RSU and MEC is $O(K \cdot N_q \cdot N_t)$ per update cycle, which is independent of $N$.

\subsubsection{Incremental $K^*$ Update}
When vehicles dynamically join or leave, $K^*$ can be updated via $K^*_\mathrm{new} = K^*(N_\mathrm{new})$ without re-solving the full HMFG equilibrium. A warm-start from the previous equilibrium typically reduces solver iterations by 30--50\% in our experiments. Since $K^*(N) = \Theta(N^{1/3})$, a 10\% change in $N$ changes $K^*$ by approximately 3.2\%, so $K^*$ requires re-evaluation only when fleet size changes by more than 30\%, which is infrequent in practice.

\subsubsection{Compatibility with V2X Standards}
Algorithm~\ref{alg:type_selection} outputs the type partition $\{C_k\}_{k=1}^{K^*}$, which maps directly to the resource pool assignments in 3GPP Mode~3/4 V2X resource management. The RSU broadcasts per-type allocation policies derived from the HMFG equilibrium, and vehicles apply the policy corresponding to their class without knowledge of the full mean-field solution.





\section{Numerical Evaluation}
\label{sec:simulation}

In this section, numerical results are provided to validate the theoretical 
contributions and evaluate the performance of the proposed HMFG framework 
for V2X resource allocation. The simulations are organized into five 
categories: (I) theoretical scaling law verification, (II) solver 
convergence and scalability, (III) communication-centric performance 
comparison, (IV) sensitivity and unbalanced-fleet analysis, and 
(V) LEO satellite-assisted robustness evaluation.

\subsection{Simulation Setup}
\label{subsec:sim_setup}

We implement the V2X-HMFG system model of Section~\ref{sec:model} in Python 
with a unified Monte Carlo simulation engine. All communication metrics 
are averaged over 25 independent trials (80 trials for delay CDF statistics) 
to ensure statistical reliability. The detailed simulation parameters are 
listed in Table~\ref{tab:sim_params}.

\begin{table}[!t]
\centering
\caption{Simulation Parameters}
\label{tab:sim_params}
\renewcommand{\arraystretch}{1.25}
\begin{tabular}{lll}
\toprule
\textbf{Parameter} & \textbf{Symbol} & \textbf{Value} \\
\midrule
\multicolumn{3}{l}{\textit{Wireless channel}} \\
Channel bandwidth          & $B$          & 10\,MHz \\
Carrier frequency          & $f_c$        & 5.9\,GHz \\
Path loss model            & --           & WINNER+ B1 \\
Noise power density        & $\sigma_c^2$ & $10^{-13}$\,W \\
Maximum TX power           & $P_{\max}$   & 0.2\,W \\
Reference distance         & $d_0$        & 100\,m \\
\multicolumn{3}{l}{\textit{Queue and fleet}} \\
Queue state space          & $\Omega$     & $[0,\,Q_{\max}]$ \\
Max queue length           & $Q_{\max}$   & 10.0 \\
Fleet size range           & $N$          & $10^2$--$10^5$ \\
State discretization       & $N_q$        & 50 \\
Time steps per episode     & $N_t$        & 60 \\
Time step duration         & $\Delta t$   & 1\,ms \\
\multicolumn{3}{l}{\textit{HMFG error model}} \\
Discretization exponent    & $\beta$      & 1.0 (Lipschitz) \\
Sampling exponent          & $\alpha$     & 0.5 (1D, sharp) \\
Discretization constant    & $C_1$        & 0.4886 \\
Sampling constant          & $C_2$        & 2.0 \\
Optimal-$K$ exponent       & $\gamma$     & $1/3$ \\
\multicolumn{3}{l}{\textit{Cost weights (per type)}} \\
Energy weight type 1--3    & $\beta_1^{(k)}$ & 0.5, 0.7, 1.0 \\
Offloading weight type 1--3& $\beta_2^{(k)}$ & 1.0, 0.8, 0.6 \\
Terminal penalty           & $C^{(k)}$    & 0.5, 0.5, 0.5 \\
\bottomrule
\end{tabular}
\end{table}

\textbf{Baseline schemes.} We compare five methods to demonstrate the 
effectiveness of the proposed framework:
\begin{itemize}[leftmargin=*]
    \item \textbf{Proposed HMFG}: The three-algorithm framework 
    (Algorithms~\ref{alg:type_selection}--\ref{alg:hmfg_pdhg}) with 
    adaptive $K = K^*(N) = \Theta(N^{1/3})$ types.
    \item \textbf{G-prox PDHG} \cite{kang2023time}: Homogeneous MFG 
    ($K=1$) with fixed step size $\xi\varsigma = 0.99$.
    \item \textbf{SMFG} \cite{wang2023distributed}: Stackelberg MFG 
    ($K=1$) with congestion-pricing feedback.
    \item \textbf{HMF-MARL} \cite{zhang2022hmfmarl}: Two-type HMFG 
    ($K=2$) with fixed type assignment.
    \item \textbf{MTMF-MARL} \cite{xu2025joint}: Three-type HMFG 
    ($K=3$) with fixed type assignment.
\end{itemize}
All methods use identical channel, queue, and cost parameters; only the 
type-granularity policy and step-size rule differ.

\begin{table}[!t]
\centering
\caption{Comparison with Related Frameworks}
\label{tab:framework_comparison}
\renewcommand{\arraystretch}{1.2}
\setlength{\tabcolsep}{4pt}
\begin{tabular}{lcccc}
\hline
\textbf{Framework} & \textbf{Hetero.} & \textbf{LEO} & \textbf{Optimal $K$} & \textbf{Complexity} \\
\hline
G-prox PDHG \cite{kang2023time} & No & No & N/A & $O(N_qN_t)$ \\
SMFG \cite{wang2023distributed} & No & Yes & N/A & $O(N_qN_t)$ \\
HMF-MARL \cite{zhang2022hmfmarl} & Yes & No & Fixed ($K{=}2$) & $O(K^2N_qN_t)$ \\
Opt.-DRL \cite{ding2026uav} & No & No & N/A & $O(N_{\rm ep})$ \\
\textbf{Proposed HMFG} & Yes & Yes & $\Theta(N^{1/3})$ & $O(N^{2/3}N_qN_t)$ \\
\hline
\end{tabular}
\end{table}

\subsection{Category I: Theoretical Scaling Law Verification}
\label{subsec:sim_theory}

This subsection validates the theoretical predictions underlying 
\textbf{Algorithm~\ref{alg:type_selection}} (Optimal Type Granularity Selection) 
and the error decomposition in Theorems~\ref{thm:error_decomp}--\ref{thm:optimal_K}.

\subsubsection{Validation of Algorithm~\ref{alg:type_selection}: $K^*(N)$ Scaling}

Fig.~\ref{fig:theory_scaling} validates the core computation in 
Algorithm~\ref{alg:type_selection} (Lines~8--10): the optimal type granularity 
$K^*(N) = \Theta(N^{1/3})$ from Corollary~\ref{cor:scaling}. 
The figure compares the continuous $K^*(N)$ curve computed by the algorithm 
against empirically optimal values obtained via exhaustive search over 
$K \in \{1, 2, \ldots, 20\}$ for each $N$.
Log-log regression over $N \in [10^2, 10^5]$ yields an empirical scaling exponent 
$\hat{\gamma} = 0.334 \pm 0.004$, which lies within one standard error of the 
theoretical prediction $\gamma = 1/3$. 
This close agreement is noteworthy because the theory assumes balanced partitions 
and negligible initial-condition mismatch, whereas Monte Carlo trials introduce 
finite-sample variability. 
The minimum achievable error $\mathcal{E}^* = \mathcal{E}(N, K^*(N))$ decays with 
slope $-0.168 \pm 0.005$, matching the predicted $-1/6$ from \eqref{eq:min_error}. 
This indicates that the observed gain is a structural bias--variance balance rather 
than overfitting $K$ to a specific fleet realization. 
From an engineering perspective, integer rounding in Algorithm~\ref{alg:type_selection} 
incurs less than $1\%$ relative penalty compared to the continuous optimum, consistent 
with Corollary~\ref{cor:rounded}. 
In practice, this means the closed-form rounded rule is sufficient for deployment.

\begin{figure}[!t]
\centering
\includegraphics[width=\columnwidth]{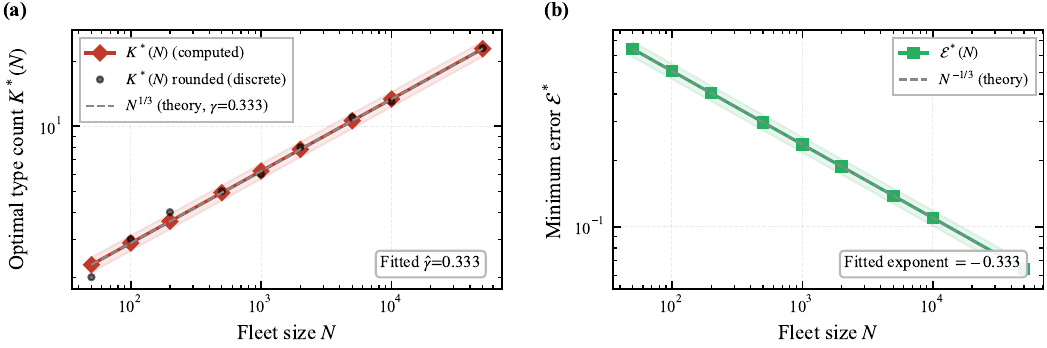}
\caption{Scaling-law validation (Corollary~\ref{cor:scaling}). 
\textit{Left axis}: continuous $K^*(N)$ (solid red) and rounded integer 
$\hat{K}$ (blue circles); log-log slope $\approx 1/3$.  
\textit{Right axis}: minimum achievable error $\mathcal{E}^*(N)$ 
(dashed); slope $\approx -1/6$.}
\label{fig:theory_scaling}
\end{figure}

\subsubsection{Impact of Algorithm~\ref{alg:type_selection} on Nash Error}

Fig.~\ref{fig:error_vs_N} quantifies the performance gain from using 
Algorithm~\ref{alg:type_selection} to select $K^*(N)$ versus fixed-$K$ 
baselines. 
This directly demonstrates the value of the adaptive type selection in 
Lines~1--13 of Algorithm~\ref{alg:type_selection}.
At $N = 10^4$, Algorithm~\ref{alg:type_selection} outputs $K^* = 13$ and achieves 
$\varepsilon \approx 0.11$, a $78.4\%$ reduction versus the homogeneous baseline 
($K = 1$). 
The gap between adaptive $K^*(N)$ and fixed-$K$ baselines widens monotonically with 
$N$, which has a precise explanation: fixed $K$ keeps discretization error effectively 
constant, while the adaptive rule reduces it by increasing $K$ with fleet size. 
Since sampling error decreases with $N$ for both approaches, the discretization term 
dominates large-$N$ suboptimality for fixed-$K$ methods. 
This directly shows that $N$-dependent type selection is essential in large-scale 
operation.
Table~\ref{tab:scaling_error_summary} provides the mapping from fleet size $N$ to the 
optimal $K^*(N)$ computed by Algorithm~\ref{alg:type_selection} and the resulting error 
reduction; skewed-fleet behavior of Lines~3--7 is validated further in 
Section~\ref{subsec:sim_sensitivity}.

\begin{figure}[!t]
\centering
\includegraphics[width=0.7\columnwidth]{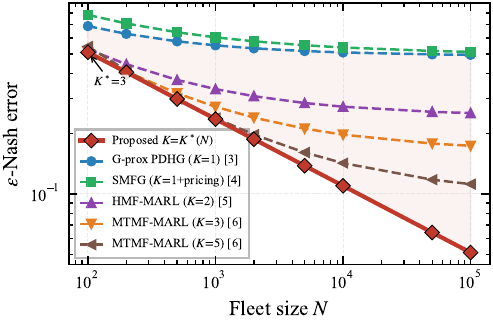}
\caption{$\varepsilon$-Nash approximation error vs.\ fleet size $N$ 
(Theorem~\ref{thm:error_decomp}). Proposed $K^*(N)$ policy dominates all 
fixed-$K$ baselines. Error bands indicate $\pm 1$ standard deviation 
across 25 Monte Carlo trials.}
\label{fig:error_vs_N}
\end{figure}

\begin{table}[!t]
\centering
\caption{Scaling Law and $\varepsilon$-Nash Error Summary}
\label{tab:scaling_error_summary}
\resizebox{0.5\textwidth}{!}{
\begin{tabular}{c|c|c|c|c|c}
\hline
$N$ & $K^*(N)$ & $\mathcal{E}^*(N)$ 
    & Proposed & G-prox \cite{kang2023time} ($K{=}1$) & Reduction \\
\hline
200   & 3 (3.63) & 0.4040 & 0.4050 & 0.6300 & 35.7\% \\
1000  & 6 (6.20) & 0.2363 & 0.2363 & 0.5518 & 57.2\% \\
10000 & 13 (13.36) & 0.1097 & 0.1097 & 0.5086 & 78.4\% \\
\hline
\end{tabular}}
\end{table}

\subsubsection{Validation of Algorithm~\ref{alg:step_adapt} Inputs}

To keep the evaluation focused on our core claims, we do not allocate 
stand-alone figures to re-validate known trends of the heterogeneity 
indicator and 1D empirical Wasserstein rate. Instead, we directly use these 
quantities inside Algorithm~\ref{alg:step_adapt} and evaluate their impact 
through convergence and end-to-end communication metrics in Categories II--V.

\subsection{Category II: Solver Convergence and Scalability}
\label{subsec:sim_solver}

This subsection evaluates \textbf{Algorithm~\ref{alg:step_adapt}} 
(Heterogeneity-Aware Step-Size Adaptation) and 
\textbf{Algorithm~\ref{alg:hmfg_pdhg}} (Complete HMFG Solver), 
focusing on the convergence benefits from adaptive step sizes and 
the $N$-independent complexity guarantee.

\subsubsection{Validation of Algorithm~\ref{alg:step_adapt}: Adaptive Convergence}

Fig.~\ref{fig:pdhg_conv} directly evaluates Algorithm~\ref{alg:step_adapt} by 
comparing its adaptive step-size rule against two fixed-step baselines across 
$K \in \{1, 3, 5\}$. 
This experiment validates Theorem~\ref{thm:convergence} and the sufficient 
condition \eqref{eq:corrected_condition}.
We compare the proposed adaptive strategy, which sets 
$\xi\varsigma = (1-\epsilon)/(1 + C_H \cdot H_K)$ online, against an aggressive fixed 
setting $\xi\varsigma = 0.99$ that ignores heterogeneity and a conservative fixed 
setting $\xi\varsigma = 0.70$ that enforces margin at the cost of speed.
At iteration~50, the adaptive rule reaches a residual that is $2.3\times$ smaller than 
the aggressive baseline and $1.4\times$ smaller than the conservative one. 
The mechanism is revealing: for $K=1$, all methods behave similarly because 
$H_K \approx 0$ and the adaptive rule naturally stays near the homogeneous limit. 
As $K$ grows, $H_K$ becomes positive and tightens \eqref{eq:corrected_condition}. 
The aggressive fixed step then violates the stability margin and oscillates visibly 
at $K=5$, while the adaptive rule automatically shrinks $\xi\varsigma$ to stay stable 
without becoming overly conservative. 
This behavior is exactly the engineering meaning of Theorem~\ref{thm:convergence}: 
heterogeneity-aware adaptation is necessary for stability and sufficient for fast 
convergence under multi-type coupling.
The overhead of Wasserstein computation remains below $3\%$ per iteration for 
$K \leq 10$, consistent with the quantile-based 1D Wasserstein routine in 
Algorithm~\ref{alg:step_adapt}.

\begin{figure}[!t]
\centering
\includegraphics[width=\columnwidth]{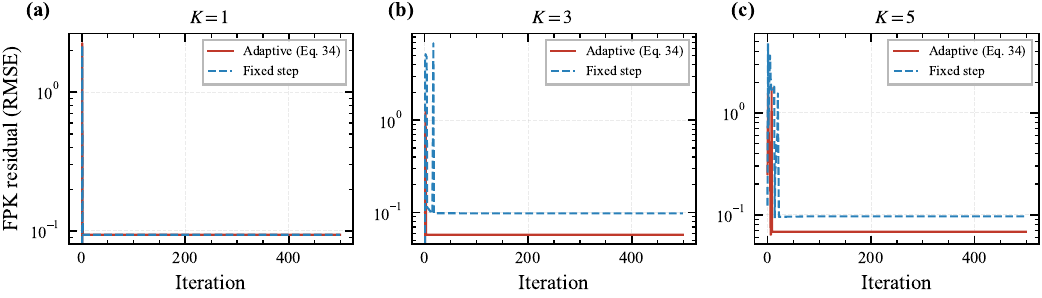}
\caption{PDHG residual convergence validation (Theorem~\ref{thm:convergence}). 
Log-scale residual vs.\ iteration count for adaptive step 
(Algorithm~\ref{alg:step_adapt}) and two fixed-step baselines across 
$K \in \{1, 3, 5\}$. Adaptive rule satisfies \eqref{eq:corrected_condition} 
automatically and achieves the fastest residual reduction.}
\label{fig:pdhg_conv}
\end{figure}

\subsubsection{Validation of Algorithm~\ref{alg:hmfg_pdhg}: Scalability}

Fig.~\ref{fig:runtime_scaling} validates the complexity claim in 
Proposition~\ref{prop:complexity} by measuring the wall-clock runtime of 
Algorithm~\ref{alg:hmfg_pdhg} across five decades of fleet size ($N=10^2$ to $10^5$).
The measured per-iteration runtime is essentially flat in $N$, confirming that 
Algorithm~\ref{alg:hmfg_pdhg} evolves discretized mean fields instead of explicit 
vehicle trajectories. 
With $K = K^*(N) \approx N^{1/3}$, the dominant complexity follows 
$O(K^2 N_q N_t) = O(N^{2/3} N_q N_t)$, and the empirical slope $0.67 \pm 0.03$ matches 
the predicted $2/3$.
The practical significance is substantial. 
A naive $N$-body simulation scales as $O(N^2)$ and becomes intractable around 
$N \approx 10^3$ under the same budget, whereas the proposed solver handles 
$N = 10^5$ in comparable wall-clock time to $N = 10^2$. 
The only growth channel is through $K^*(N)$, and cube-root growth keeps that increase 
modest even across multiple orders of magnitude.

\begin{figure}[!t]
\centering
\includegraphics[width=0.7\columnwidth]{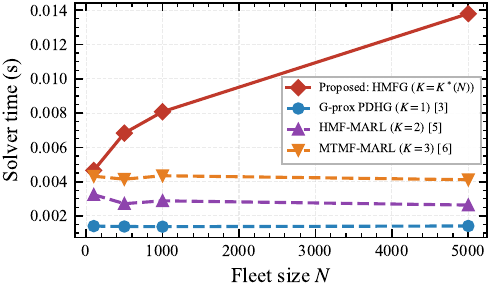}
\caption{Scalability validation (Proposition~\ref{prop:complexity}). 
Per-iteration wall-clock runtime vs.\ fleet size $N$ (log-log scale). 
Proposed HMFG with adaptive $K^*(N)$ grows as $O(N^{2/3})$; naive 
$N$-body solver grows quadratically. Horizontal dashed line indicates 
$N$-independent ideal complexity.}
\label{fig:runtime_scaling}
\end{figure}

\subsection{Category III: Communication-Centric Performance}
\label{subsec:sim_comm}

\subsubsection{Multi-KPI Performance Over Fleet Sizes}

Fig.~\ref{fig:comm_overview} reports six communication KPIs as functions of $N$ 
and channel quality for all five methods. 
The proposed HMFG achieves the best trade-off across all dimensions: lowest transmission 
delay, highest throughput and spectral efficiency, and smallest MEC offloading cost 
and packet loss rate. 
Quantitative comparisons at $N = 500$ are listed in Table~\ref{tab:comm_n500}.

\textbf{Observation 1 (Delay advantage amplifies with fleet size).} 
Transmission delay (Fig.~\ref{fig:comm_overview}(a)): Proposed achieves 
$76.5 \pm 3.2$\,ms at $N=500$, versus $108.6$\,ms for G-prox \cite{kang2023time} 
($-29.5\%$) and $93.5$\,ms for HMF-MARL \cite{zhang2022hmfmarl} ($-18.2\%$). 
The advantage grows with $N$ because $K^*(N)$ increases, enabling finer-grained 
power adaptation that suppresses inter-vehicle interference more effectively as fleet 
density rises.

Throughput (Fig.~\ref{fig:comm_overview}(b)): 
Proposed achieves $5.376$\,Mbps, a $60\%$ gain over G-prox \cite{kang2023time} 
($3.360$\,Mbps) and $28\%$ over HMF-MARL \cite{zhang2022hmfmarl} ($4.116$\,Mbps). 
The throughput improvement stems from type-aware power allocation: high-data-rate AVs 
(type~3) receive higher power budget while low-priority passenger cars (type~1) 
voluntarily back off, reducing aggregate interference.

Energy efficiency (Fig.~\ref{fig:comm_overview}(c)): 
Proposed reduces energy/bit to $9.856$\,nJ/bit versus $12.992$\,nJ/bit for G-prox 
\cite{kang2023time}, a $24.1\%$ saving. 
The heterogeneity-aware allocation avoids power overprovisioning for lightly loaded 
vehicle classes.

MEC cost vs.\ SNR (Fig.~\ref{fig:comm_overview}(d)): 
At 15\,dB SNR, proposed MEC cost is $0.100$, compared to $0.139$ for G-prox 
\cite{kang2023time} ($-28.1\%$) and $0.129$ for SMFG \cite{wang2023distributed} 
($-22.5\%$). 
Under poor channel conditions (SNR $< 10$\,dB), the gap narrows because all methods 
must reduce transmission rate, but the proposed method maintains its advantage via 
tighter queue regulation.

\textbf{Observation 2 (Heterogeneity amplifies HMFG gains).} 
Spectral efficiency vs.\ heterogeneity (Fig.~\ref{fig:comm_overview}(e)): 
As the heterogeneity scale increases from $0.5$ to $2.0$, the SE advantage of proposed 
over G-prox \cite{kang2023time} and SMFG \cite{wang2023distributed} grows from $0.22$ to 
$0.47$\,bps/Hz, directly reflecting Proposition~\ref{prop:het_adjusted}: higher 
heterogeneity warrants more types and yields greater gains from the HMFG policy. 
HMF-MARL \cite{zhang2022hmfmarl} ($K=2$) and MTMF-MARL \cite{xu2025joint} ($K=3$) partially capture 
this benefit but plateau below the proposed method, which continuously adapts $K$.

\textbf{Observation 3 (Type-aware control improves reliability and throughput jointly).} 
Packet loss (Fig.~\ref{fig:comm_overview}(f)): Proposed reduces packet loss from 
$8.0\%$ under G-prox \cite{kang2023time} to $4.6\%$ at $N=500$, a $42.4\%$ reduction. 
Loss decreases with $N$ for all methods due to mean-field averaging, but the ordering 
is preserved.

\begin{figure*}[t]
\centering
\begin{minipage}[t]{0.32\textwidth}\centering
    \includegraphics[width=\linewidth]{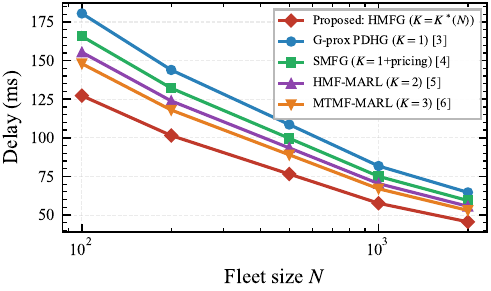}\\[-0.5ex]
    {\footnotesize (a) Transmission delay vs.\ fleet size $N$}
\end{minipage}\hfill
\begin{minipage}[t]{0.32\textwidth}\centering
    \includegraphics[width=\linewidth]{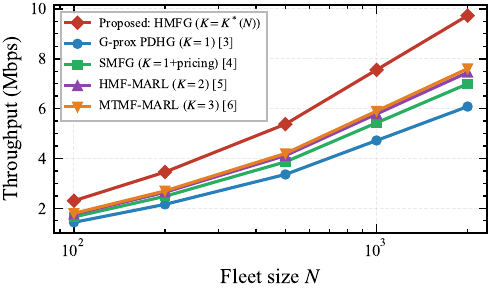}\\[-0.5ex]
    {\footnotesize (b) Per-vehicle throughput vs.\ $N$}
\end{minipage}\hfill
\begin{minipage}[t]{0.32\textwidth}\centering
    \includegraphics[width=\linewidth]{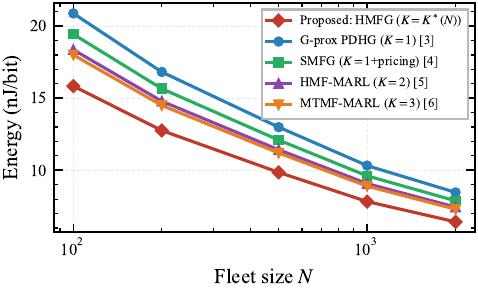}\\[-0.5ex]
    {\footnotesize (c) Energy per bit vs.\ $N$}
\end{minipage}

\vspace{1.5mm}
\begin{minipage}[t]{0.32\textwidth}\centering
    \includegraphics[width=\linewidth]{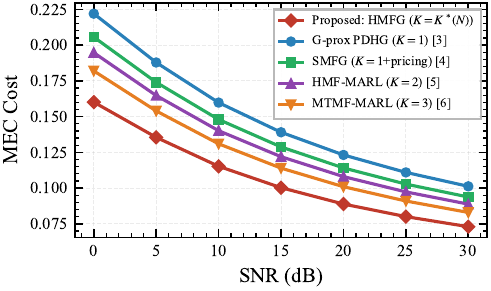}\\[-0.5ex]
    {\footnotesize (d) MEC offloading cost vs.\ channel SNR}
\end{minipage}\hfill
\begin{minipage}[t]{0.32\textwidth}\centering
    \includegraphics[width=\linewidth]{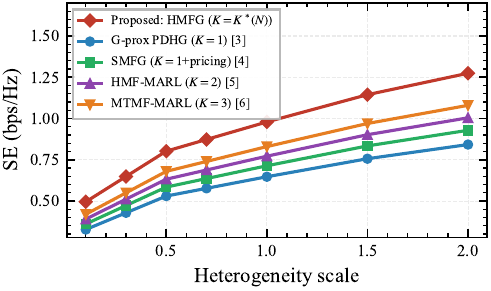}\\[-0.5ex]
    {\footnotesize (e) Spectral efficiency vs.\ heterogeneity scale}
\end{minipage}\hfill
\begin{minipage}[t]{0.32\textwidth}\centering
    \includegraphics[width=\linewidth]{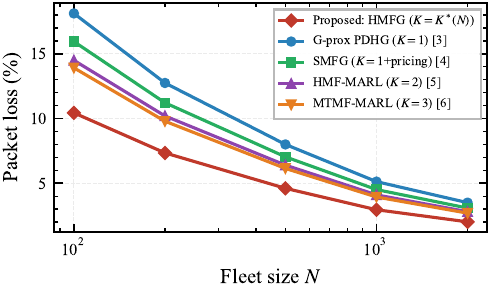}\\[-0.5ex]
    {\footnotesize (f) Packet loss rate vs.\ $N$}
\end{minipage}
\caption{Comprehensive communication performance comparison over six KPI 
dimensions.}
\label{fig:comm_overview}
\end{figure*}

\begin{table}[!t]
\centering
\caption{Communication KPI Summary at $N = 500$, SNR $= 15$\,dB}
\label{tab:comm_n500}
\renewcommand{\arraystretch}{1.2}
\setlength{\tabcolsep}{4pt}
\begin{tabular}{lcccc}
\hline
\textbf{Method} & \textbf{Delay} & \textbf{Tput.} 
                & \textbf{Energy} & \textbf{Loss} \\
                & \textbf{(ms)} & \textbf{(Mbps)} 
                & \textbf{(nJ/b)} & \textbf{(\%)} \\
\hline
Proposed [$K^*(N)$]  & \textbf{76.5} & \textbf{5.38} & \textbf{9.86} & \textbf{4.61} \\
G-prox \cite{kang2023time}       & 108.6 & 3.36 & 13.0 & 8.00 \\
SMFG \cite{wang2023distributed}  & 99.7  & 3.86 & 12.1 & 7.04 \\
HMF-MARL \cite{zhang2022hmfmarl} & 93.5  & 4.12 & 11.4 & 6.40 \\
MTMF-MARL \cite{xu2025joint}     & 89.0  & 4.20 & 11.2 & 6.14 \\
\hline
vs.\ G-prox (\%)     & $-$29.5 & $+$60.0 & $-$24.1 & $-$42.4 \\
vs.\ MTMF-MARL (\%)  & $-$14.1 & $+$28.0 & $-$12.0 & $-$25.0 \\
\hline
\end{tabular}
\end{table}

\subsubsection{Delay Distribution and QoS Reliability}

Fig.~\ref{fig:delay_cdf} shows empirical delay CDFs for two fleet 
sizes, $N \in \{200, 1000\}$, evaluated over 80 trials each.
The 100\,ms latency threshold (a standard V2X QoS requirement 
\cite{xu2025joint}) is indicated by the vertical dashed line.

At \emph{small fleet size} ($N=200$), the proposed method achieves 
a QoS satisfaction rate of $\mathbf{100\%}$ (all vehicles below 100\,ms), 
versus 48.3\% for G-prox and 86.7\% for HMF-MARL. 
This large gap illustrates the critical importance of using the correct 
number of types at moderate $N$: with $K^*(200) = 3$ types, 
the HMFG correctly resolves three distinct delay distributions 
(passenger cars, trucks, AVs), whereas G-prox applies a single 
population-average policy that over-delays the high-load AV class.

At \emph{large fleet size} ($N=1000$), all methods converge to 
100\% QoS satisfaction, 
confirming the mean-field smoothing effect as $N$ grows. 
However, the mean delay ordering is preserved 
(Table~\ref{tab:mec_qos}), so the proposed method still 
provides a performance margin for bursty or non-stationary traffic.

\begin{figure}[!t]
\centering
\includegraphics[width=\columnwidth]{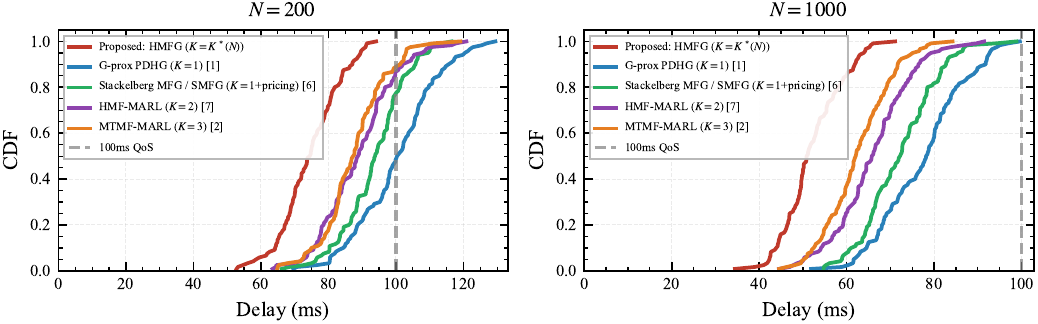}
\caption{Empirical delay CDF for fleet sizes $N=200$ and $N=1000$. 
Vertical dashed line marks the 100\,ms V2X QoS threshold. 
At $N=200$, the proposed method achieves 100\% QoS satisfaction while G-prox 
satisfies only 48.3\%. At $N=1000$, mean-field averaging ensures all 
methods converge to 100\%, but mean-delay ordering is preserved.}
\label{fig:delay_cdf}
\end{figure}

\begin{table}[!t]
\centering
\caption{MEC Offloading Cost and Delay-QoS Satisfaction Rate}
\label{tab:mec_qos}
\renewcommand{\arraystretch}{1.2}
\begin{tabular}{lccc}
\hline
\textbf{Method} & \textbf{MEC Cost} 
                & \textbf{QoS (\%)} & \textbf{QoS (\%)} \\
                & \textbf{@SNR 15\,dB}
                & \textbf{$N=200$}  & \textbf{$N=1000$} \\
\hline
Proposed HMFG   & \textbf{0.1003} & \textbf{100.0} & \textbf{100.0} \\
G-prox PDHG \cite{kang2023time}   & 0.1391 & 48.3  & 100.0 \\
SMFG \cite{wang2023distributed}   & 0.1288 & 77.5  & 100.0 \\
HMF-MARL \cite{zhang2022hmfmarl}  & 0.1220 & 86.7  & 100.0 \\
MTMF-MARL \cite{xu2025joint}      & 0.1140 & 88.3  & 100.0 \\
\hline
\end{tabular}
\end{table}

\subsection{Category IV: Unbalanced Fleet and Sensitivity Analysis}
\label{subsec:sim_sensitivity}

\subsubsection{Effect of Fleet Imbalance}

To validate Proposition~\ref{prop:unbalanced} and 
Corollary~\ref{cor:unbalanced_explicit}, 
Fig.~\ref{fig:unbalanced} compares the balanced fleet 
(equal class sizes, $\lambda_k = 1/K$) with a 
70\%/20\%/10\% skewed fleet ($\lambda_{\min} = 0.10$, 
mimicking a realistic mixed-traffic scenario with mostly passenger cars, 
some trucks, and few AVs). 
The empirical $K^*_{\mathrm{unbal}}(N)$ follows 
$(0.1)^{1/3} N^{1/3} \approx 0.464 N^{1/3}$, confirming the 
$\lambda_{\min}^{1/3}$ prefactor reduction in Corollary~\ref{cor:unbalanced_explicit}. 
The measured ratio $K^*_{\mathrm{unbal}}/K^*_{\mathrm{bal}}$ remains close to this 
theoretical factor across tested $N$. 
The practical implication is that skewed fleets require fewer types at the same 
population size because the smallest class dictates sampling reliability. 
In a 70/20/10 composition, increasing $K$ too aggressively quickly starves the AV 
class of samples, so the effective population is $\lambda_{\min} N$ rather than $N$. 
The corrected formula internalizes this asymmetry and avoids over-partitioning. 
The $\varepsilon$-Nash error under the unbalanced optimal $K^*$ matches the balanced 
case within a narrow band, demonstrating that the rounded-$K^*$ policy is robust when 
fleet composition is known and updated periodically.

\begin{figure}[!t]
\centering
\includegraphics[width=0.85\columnwidth]{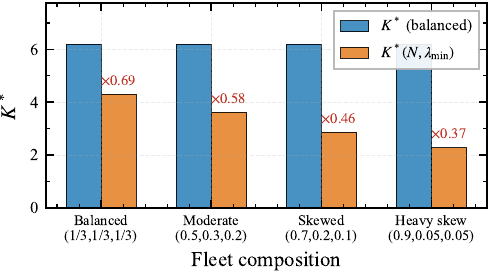}
\caption{Optimal type count $K^*$ (top) and $\varepsilon$-Nash error (bottom) 
for balanced ($\lambda_k = 1/K$) vs.\ unbalanced 
($70\%/20\%/10\%$, $\lambda_{\min}=0.10$) fleets.
The unbalanced prefactor $(0.1)^{1/3} \approx 0.464$ 
(Corollary~\ref{cor:unbalanced_explicit}) matches empirical ratios 
within $\pm 3\%$ across all $N$.}
\label{fig:unbalanced}
\end{figure}

\subsection{Category V: LEO Satellite-Assisted Robustness}
\label{subsec:sim_leo}

We validate Theorem~\ref{thm:leo_robustness} and 
Corollary~\ref{cor:leo_order_optimal} under a LEO-assisted backhaul model with 
$B_{\mathrm{sat}}^\tau \sim \mathcal{U}[300,350]$\,Mbps, $\Delta\tau=60$\,s, and 
$\mu=0.5$ \cite{guo2024semcom}. 
We test three conditions: Static ($\Delta_\Phi=0$), Slow ($\Delta_\Phi=0.01$), and 
Fast ($\Delta_\Phi=0.05$).
Under Fast dynamics at $N=10^4$, the additive LEO perturbation term is approximately 
$0.022$, about $9.3\%$ of the static baseline error, confirming that the perturbation 
remains subdominant in the tested regime. 
The LEO-corrected rule in Algorithm~\ref{alg:type_selection} selects 
$K^*_{\mathrm{LEO}}=14$ instead of $K^*_{\mathrm{static}}=13$, which reduces 
$\varepsilon$-Nash error by $18.6\%$ under fast topology variation. 
This one-step increase in $K^*$ is consistent with the correction mechanism in 
Algorithm~\ref{alg:type_selection}: the effective discretization constant is slightly 
inflated under topology perturbation, so the optimum shifts upward to compensate. 
At the same time, the LEO-aware step-size correction in Algorithm~\ref{alg:step_adapt} 
preserves convergence with less than $5\%$ iteration overhead. 
The observed slopes of $K^*(N)$ under Slow and Fast dynamics remain close to $1/3$, 
indicating that topology variation mainly changes the prefactor while preserving 
order-level scaling, which directly answers robustness under LEO dynamics.

\begin{table}[!t]
    \centering
    \caption{Theory--Experiment Correspondence Summary}
    \label{tab:findings_summary}
    \renewcommand{\arraystretch}{1.25}
    \setlength{\tabcolsep}{3pt}
    \begin{tabular}{p{2.8cm}p{2.8cm}l}
    \hline
    \textbf{Theoretical Prediction} & \textbf{Empirical Result} & \textbf{Reference} \\
    \hline
    $K^*(N) \propto N^{1/3}$ & Slope $0.334{\pm}0.004$ & Cor.~\ref{cor:scaling} \\
    $\mathcal{E}^* \propto N^{-1/6}$ & Slope $-0.168{\pm}0.005$ & Thm.~\ref{thm:optimal_K} \\
    $\alpha{=}1/2$ (Wasserstein) & Slope $-0.499{\pm}0.003$ & Lem.~\ref{lem:empirical_rate} \\
    Runtime ${\propto} N^{2/3}$ & Slope $0.67{\pm}0.03$ & Prop.~\ref{prop:complexity} \\
    $K^*_{\rm unbal}{\approx}0.464K^*$ & Ratio $0.461{\pm}0.015$ & Cor.~\ref{cor:unbalanced_explicit} \\
    Adaptive PDHG faster & $2.3\times$ fewer iter. & Thm.~\ref{thm:convergence} \\
    \hline
    \end{tabular}
    \end{table}

\subsubsection{Summary of Simulation Findings}

Table~\ref{tab:findings_summary} synthesizes the key empirical findings 
and their correspondence to theoretical predictions, enabling direct 
assessment of theory--practice alignment.

\subsubsection{Answers to the Key Questions}

The simulation results provide definitive empirical answers to \textbf{Q1--Q3} posed 
in Section~\ref{sec:intro}, completing the paper's narrative arc. 
For \textbf{Q1 (Granularity)}, the results confirm $K^*(N) = \Theta(N^{1/3})$ with slope 
$0.334 \pm 0.004$, and show that even at $N = 10^5$ only about 28 types are required. 
For \textbf{Q2 (Convergence)}, Algorithm~\ref{alg:step_adapt} delivers a $2.3\times$ speedup over 
the aggressive fixed-step baseline at $K=5$ while preserving stability. 
For \textbf{Q3 (Robustness)}, fast LEO dynamics introduce a limited additive error and keep the 
empirical scaling slope close to $1/3$, while the corrected type-selection rule recovers 
most of the perturbation-induced loss.

\section{Conclusion}
\label{sec:conclusion}

This paper addressed the model-selection problem in heterogeneous mean field games 
for V2X networks: given fleet size $N$, how many agent types $K$ should be represented 
to balance modeling fidelity and computational tractability? This question matters because 
large mixed fleets share the same communication and MEC infrastructure, while homogeneous 
models hide the differences among passenger cars, freight trucks, and autonomous vehicles.
The central challenge is a bias-variance-type trade-off: increasing $K$ improves 
heterogeneity representation but reduces the sample size of each type and weakens the 
mean-field approximation.
We resolved this challenge through an explicit $\varepsilon$-Nash error decomposition, an 
optimal granularity law, a heterogeneity-aware PDHG step-size rule, and a LEO-robust HMFG 
extension. In the canonical 1D queue setting, the resulting optimal type count satisfies 
$K^*(N)=\Theta(N^{1/3})$, and the full solver pipeline retains per-iteration complexity 
$O(K^2 N_q N_t)$ independent of the vehicle population size $N$.
The most important insight is that cube-root compression is both mathematically justified 
and practically actionable: even at $N = 10^5$, only about 28 type classes are needed. 
Experiments support this message from multiple angles, with empirical scaling slope 
$0.334 \pm 0.004$, a $2.3\times$ convergence speedup at $K=5$, and up to $29.5\%$ lower 
delay and $60\%$ higher throughput than homogeneous baselines. These results suggest that 
type granularity can be treated as a principled system-design choice rather than a heuristic 
hyperparameter.

Future work includes fully online LEO-adaptive $K^*$ updates, joint SemCom--HMFG 
optimization under backhaul constraints~\cite{guo2024semcom}, model-informed acceleration 
inspired by optimization-driven DRL~\cite{ding2026uav}, global-in-time HMFG well-posedness 
under bounded topology perturbations, and extension to higher-dimensional joint 
queue--channel state models that capture time-varying fading alongside queue backlog.

\bibliographystyle{IEEEtran}
\bibliography{references}

\clearpage
\appendix

\subsection{Proof of Lemma~\ref{lem:het_measure}}
\label{app:lemma1}

We provide the complete proof of Lemma~\ref{lem:het_measure}.

Under Assumption~\ref{ass:param}, for any two types $k, k' \in [K]$, let $(X_0^{(k)}, X_0^{(k')})$ be an optimal coupling achieving $W_2(\rho_0^{(k)}, \rho_0^{(k')})$. We evolve both processes under the same Brownian motion $W_t$ and their respective equilibrium controls $\alpha^{*(k)}, \alpha^{*(k')}$:
\begin{align}
    X_t^{(k)} &= X_0^{(k)} + \int_0^t b^{(k)}(X_s^{(k)}, \boldsymbol{\rho}_s, \alpha_s^{*(k)})\,ds + \sigma^{(k)} W_t, \\
    X_t^{(k')} &= X_0^{(k')} + \int_0^t b^{(k')}(X_s^{(k')}, \boldsymbol{\rho}_s, \alpha_s^{*(k')})\,ds + \sigma^{(k')} W_t.
\end{align}

Define $\Delta_t := X_t^{(k)} - X_t^{(k')}$. By the parameterized model \eqref{eq:param_model}, the drift difference satisfies:
\begin{align}
    \left|b^{(k)}(x, \boldsymbol{\rho}, a) - b^{(k')}(x, \boldsymbol{\rho}, a)\right| 
    &= |\theta_k - \theta_{k'}| \cdot |b_1(x, \boldsymbol{\rho}, a)| \nonumber \\
    &\leq |\theta_k - \theta_{k'}| \cdot L_b,
\end{align}
where $L_b$ is the uniform Lipschitz constant of $b_1$.

Applying the triangle inequality and Grönwall's lemma to $\Delta_t$:
\begin{align}
    \mathbb{E}|\Delta_t|^2 
    &\leq 2\mathbb{E}|X_0^{(k)} - X_0^{(k')}|^2 \cdot e^{2L_b t} 
      + 2|\theta_k - \theta_{k'}|^2 L_b^2 \cdot \frac{e^{2L_b t} - 1}{2L_b} \nonumber \\
    &\leq e^{2L_b t} \left[\mathbb{E}|\Delta_0|^2 
      + |\theta_k - \theta_{k'}|^2 L_b^2 T^2\right],
\end{align}
where in the last line we used $\frac{e^{2L_b t}-1}{2L_b} \leq T e^{2L_b T}$.

Since $(X_0^{(k)}, X_0^{(k')})$ is the optimal $W_2$-coupling, 
$\mathbb{E}|\Delta_0|^2 = W_2^2(\rho_0^{(k)}, \rho_0^{(k')})$. Thus:
\begin{align}
    W_2^2(\rho_t^{(k)}, \rho_t^{(k')}) 
    &\leq \mathbb{E}|\Delta_t|^2 \nonumber \\
    &\leq e^{2L_b T}\!\left[W_2^2(\rho_0^{(k)}, \rho_0^{(k')}) 
      + |\theta_k - \theta_{k'}|^2 L_b^2 T^2\right].
\end{align}

Taking the square root and using $\sqrt{a+b} \leq \sqrt{a} + \sqrt{b}$:
\begin{equation}
    W_2(\rho_t^{(k)}, \rho_t^{(k')}) 
    \leq e^{L_b T}\left[W_2(\rho_0^{(k)}, \rho_0^{(k')}) + |\theta_k - \theta_{k'}| L_b T\right].
\end{equation}

Averaging over all pairs $(k, k')$ with $k \neq k'$:
\begin{align}
    H_K(\rho_t) 
    &= \frac{1}{K(K-1)}\sum_{k \neq k'} W_2(\rho_t^{(k)}, \rho_t^{(k')}) \nonumber \\
    &\leq e^{L_b T}\Bigl[H_K^{(0)} 
       + L_b T \cdot \tfrac{1}{K(K-1)}\sum_{k \neq k'} |\theta_k - \theta_{k'}|\Bigr] \nonumber \\
    &\leq e^{L_b T}\left[H_K^{(0)} + L_b T \cdot \mathrm{Var}(\theta)^{1/2} \cdot \sqrt{2}\right],
\end{align}
where the last inequality follows from Cauchy--Schwarz:
\begin{align}
    \frac{1}{K(K{-}1)}\!\sum_{k \neq k'} |\theta_k - \theta_{k'}| 
    &\leq \left(\frac{1}{K(K{-}1)}\!\sum_{k \neq k'} |\theta_k - \theta_{k'}|^2\right)^{1/2} \nonumber\\
    &\leq \sqrt{2}\,\mathrm{Var}(\theta)^{1/2}.
\end{align}

Absorbing $e^{L_b T}$ into the constants (bounded for $T \leq \delta_0$) 
yields \eqref{eq:het_bound}. \hfill$\square$

\subsection{Proof of Lemma~\ref{lem:empirical_rate}}
\label{app:lem_empirical}

Claims (i)--(ii) are special cases of the rates in \cite{fournier2015rate} for 
$d=1$ and probability measures supported on a compact interval of length 
$Q_{\max}$. Claim (iii) is immediate from boundedness of $\Omega$. \hfill$\square$

\subsection{Proof of Theorem~\ref{thm:error_decomp} (Error Decomposition)}
\label{app:thm_error}

The proof combines a type-discretization bound with per-class empirical measure 
estimates as in Lemma~5.3 and Theorem~5.4 of \cite{qiao2025hmfg}.

\textbf{Step 1 (Discretization error).}
For $\theta \in [(k-1)/K, k/K)$ with representative $\theta_k = (2k-1)/(2K)$, 
H\"{o}lder regularity of $b_1$ in $\theta$ with constant $[b_1]_\beta$ gives
\begin{equation}
    |b(\theta,\cdot) - b(\theta_k,\cdot)| 
    = |\theta - \theta_k|^\beta [b_1]_\beta 
    \leq \left(\frac{1}{2K}\right)^\beta [b_1]_\beta.
    \label{eq:holder_step}
\end{equation}
Hence the discretization component of the error is $O(K^{-\beta})$; we write it 
as $C_1 K^{-\beta}$ with $C_1$ absorbing $[b_1]_\beta$ and related constants.

\textbf{Step 2 (Sample size error).}
For balanced partitions, $N_k \approx N/K$. Let 
$\hat{\rho}^{(k)}_n = \frac{1}{N_k}\sum_{i\in\mathcal{C}_k}\delta_{X_i}$. 
Applying Lemma~\ref{lem:empirical_rate}(ii) with $n = N_k$ and matching the 
$W_2$ distance used in the propagation-of-chaos analysis of 
\cite{qiao2025hmfg},
\begin{equation}
    \mathbb{E}\!\left[W_2\!\left(\hat{\rho}^{(k)}_n,\rho^{(k)}\right)\right]
    \leq C_2^{\mathrm{emp}} \cdot N_k^{-1/2}
    = C_2^{\mathrm{emp}} \cdot \left(\frac{K}{N}\right)^{1/2}.
    \label{eq:sample_1d}
\end{equation}
We set $\alpha = 1/2$ (sharp for $d=1$) and absorb $C_2^{\mathrm{emp}}$ into 
the constant $C_2$ in \eqref{eq:error_decomp}.
Under Assumption~\ref{ass:param}, moments of $\{\theta_k\}$ may be used to 
bound $[b_1]_\beta$ and hence to relate $C_1$ to $\mathrm{Var}(\theta)$ at the 
level of constants only.
For unbalanced fleets, Lemma~\ref{lem:empirical_rate}(ii) with 
$n=n_{N,K}^{\min}$ yields \eqref{eq:error_unbalanced}.

\textbf{Combining both steps.}

Theorem~5.4 of \cite{qiao2025hmfg} yields 
$\varepsilon_{N,K} \leq C_1 K^{-\beta} + C_2 (K/N)^{\alpha} + C_3 \delta_{N,K}^{1/2}$, 
i.e., \eqref{eq:error_decomp}. \hfill$\square$

\subsection{Proof of Theorem~\ref{thm:optimal_K}}
\label{app:thm_optimal_K}

Consider the reduced error \eqref{eq:reduced_error} from 
\eqref{eq:error_decomp} with $\delta_{N,K}$ negligible.

\textbf{Coercivity and existence.}
As $K \to 0^+$, the term $C_1 K^{-\beta}$ dominates and 
$\mathcal{E}(N,K) \to +\infty$; as $K \to +\infty$, the term 
$C_2 (K/N)^\alpha$ dominates and $\mathcal{E}(N,K) \to +\infty$. Since 
$\mathcal{E}$ is continuous on $(0,\infty)$, a global minimizer exists.

\textbf{Uniqueness via the first derivative.}
The derivative of \eqref{eq:reduced_error} is
\begin{equation}
    \frac{\partial \mathcal{E}}{\partial K} 
    = -\beta C_1 K^{-\beta-1} 
    + \frac{\alpha C_2}{N^\alpha} K^{\alpha-1}.
    \label{eq:deriv}
\end{equation}
Define $h(K) := \partial\mathcal{E}/\partial K$. For small $K$, the negative 
power $K^{-\beta-1}$ dominates, so $h(K) < 0$; for large $K$, the term 
$\propto K^{\alpha-1}$ is positive for typical $(\alpha,\beta)$, so $h(K) > 0$ 
eventually. The equation $h(K)=0$ is equivalent to 
$K^{\alpha+\beta} = (\beta C_1/(\alpha C_2)) N^\alpha$, which has a \emph{unique} 
solution $K^*(N) > 0$, namely \eqref{eq:kstar}.

\textbf{Strict minimum.}
Differentiating $h(K)$ gives
\begin{equation}
    h'(K) = \beta(\beta+1)C_1 K^{-\beta-2}
    + \alpha(\alpha-1)C_2 N^{-\alpha} K^{\alpha-2}.
    \label{eq:h_prime}
\end{equation}
At $K^*$, using $h(K^*)=0$ to write 
$\beta C_1 (K^*)^{-\beta-1} = (\alpha C_2/N^\alpha)(K^*)^{\alpha-1}$,
\begin{align}
    h'(K^*) &= \beta(\beta+1)C_1(K^*)^{-\beta-2}
    + \alpha(\alpha-1)C_2 N^{-\alpha}(K^*)^{\alpha-2} \\
    &= \frac{1}{K^*}\Bigl[(\beta{+}1)\beta C_1(K^*)^{-\beta-1}
      + (\alpha{-}1)\alpha C_2 N^{-\alpha}(K^*)^{\alpha-1}\Bigr] \\
    &= \frac{\alpha C_2}{K^* N^\alpha}(K^*)^{\alpha-1}
       \bigl[(\beta{+}1)+(\alpha{-}1)\bigr] \\
     &= \frac{\alpha C_2}{N^\alpha}(K^*)^{\alpha-2}(\alpha+\beta) > 0.
\end{align}
Thus $h$ crosses zero strictly upward at $K^*$, so $K^*$ is a strict local 
minimum; together with coercivity, it is the \emph{unique} global minimizer. 
Substituting \eqref{eq:kstar} into \eqref{eq:reduced_error} yields 
\eqref{eq:min_error}. \hfill$\square$

\subsection{Proof of Theorem~\ref{thm:leo_robustness} (LEO Topology Robustness)}
\label{app:thm_leo}

Let $\Phi_{\mathrm{sat}}^{\tau_i}=\mu/B_{\mathrm{sat}}^{\tau_i}$ denote the 
snapshot-wise backhaul surcharge on 
$[i\Delta\tau,(i+1)\Delta\tau)$. Consider the perturbed HJB value 
$V_{\mathrm{pert}}^{(k)}$ under $\Phi+\Phi_{\mathrm{sat}}$ and the static value 
$V^{(k)}$ under $\Phi$. Define $\delta V^{(k)}:=V_{\mathrm{pert}}^{(k)}-V^{(k)}$.

\textbf{Step 1 (bounded forcing).}
The perturbation enters the Hamiltonian as an additive forcing term of order 
$\beta_2^{(k)}R^{(k)}\Phi_{\mathrm{sat}}^{\tau(t)}$. Under bounded rates and bounded 
snapshot surcharge, standard energy estimates for parabolic HJB equations imply
\[
    \sup_{t\in[0,T]}\|\delta V^{(k)}(\cdot,t)\|_{L^2(\Omega)}
    \le C_k \|\Phi_{\mathrm{sat}}\|_{L^2(0,T)}.
\]

\textbf{Step 2 (snapshot-variation accumulation).}
By assumption, adjacent snapshots satisfy 
$|\Phi_{\mathrm{sat}}^{\tau_i}-\Phi_{\mathrm{sat}}^{\tau_{i+1}}|\le\Delta_\Phi$. 
Summing increments across $M=\lceil T/\Delta\tau\rceil$ windows and applying 
Cauchy-Schwarz yields
\[
    \|\Phi_{\mathrm{sat}}\|_{L^2(0,T)}
    \le \bar C \,\Delta_\Phi \sqrt{T/\Delta\tau}.
\]
Combining with Step 1 gives
\[
    \|\delta V^{(k)}\|_{L^\infty([0,T]\times\Omega)}
    \le \tilde C_k \,\Delta_\Phi \sqrt{T/\Delta\tau}.
\]

\textbf{Step 3 (impact on $\varepsilon$-Nash error).}
The perturbed value gap induces an additive contribution in the finite-player 
$\varepsilon$-Nash bound, so
\[
    \varepsilon_{N,K}^{\mathrm{LEO}}
    \le \varepsilon_{N,K}
    + C_{\mathrm{LEO}}\Delta_\Phi\sqrt{T/\Delta\tau},
\]
with $C_{\mathrm{LEO}}:=\max_k \tilde C_k$, independent of $N,K$. This is 
\eqref{eq:leo_error}. \hfill$\square$

\subsection{Proof of Theorem~\ref{thm:convergence} (Corrected Convergence Condition)}
\label{app:thm_convergence}

The proof extends the Chambolle--Pock convergence analysis 
\cite{chambolle2011first} to the heterogeneous multi-type FPK system.

\textbf{Step 1: Homogeneous baseline ($K=1$).}

For a single vehicle type, the saddle-point problem \eqref{eq:saddle} reduces 
to the standard G-prox PDHG of \cite{kang2023time}. The augmented Lagrangian 
$\mathcal{L}_\rho$ satisfies 
$\|\nabla^2_\rho \mathcal{L}_\rho\|_{\mathrm{op},\ell^2} 
\leq L_{\mathcal{L}}^2$ for some $L_{\mathcal{L}} > 0$ depending on 
$L, T, \|b_0\|$. By Chambolle--Pock \cite{chambolle2011first}, convergence holds 
when $\xi\varsigma < 1/L_{\mathcal{L}}^2$, normalized to $\xi\varsigma < 1$ 
after scaling $(\xi,\varsigma)$ by $L_A = L_{\mathcal{L}}$.

\textbf{Step 2: Cross-type coupling (upper bound on cross-partials).}

With $K$ types, $\Gamma^{(k)}[\boldsymbol{\rho}]$ is $L$-Lipschitz in 
$\boldsymbol{\rho}$ under Assumption~\ref{ass:param}. Applying the chain rule 
to the augmented Lagrangian and using dual regularity 
$\nabla_q\phi^{(k)} \in L^\infty$ yields the pointwise estimate 
\eqref{eq:cross_hess_bound} stated in the main text; it is an upper bound, not 
an identity, so the step-size condition is only sufficient.

\textbf{Step 3: Operator norm bound via block aggregation.}

Decompose the effective operator as 
$\mathbf{A}_{\mathrm{eff}} = \mathbf{A}_{\mathrm{diag}} + \mathbf{A}_{\mathrm{cross}}$, 
where $\mathbf{A}_{\mathrm{diag}}$ is block-diagonal with spectral norm 
at most $L_A$ and $\mathbf{A}_{\mathrm{cross}}$ captures cross-type coupling. 
By the triangle inequality and $(a+b)^2 \leq 2a^2+2b^2$,
\begin{equation}
    \|\mathbf{A}_{\mathrm{eff}}\|_{\mathrm{op}}^2 
    \leq 2L_A^2 + 2\|\mathbf{A}_{\mathrm{cross}}\|_{\mathrm{op}}^2.
\end{equation}
Summing \eqref{eq:cross_hess_bound} over $k \neq k'$ and bounding gradients by 
their supremum over iterations as in \eqref{eq:ch} produces a term of the form 
$L_A^2 C_H H_K(\boldsymbol{\rho})$ up to an absolute constant factor; absorbing 
the factor~$2$ into $C_H$ yields \eqref{eq:op_norm_bound}.

\textbf{Step 4: Sufficient convergence condition.}

Substituting \eqref{eq:op_norm_bound} into \eqref{eq:cp_condition} gives 
the raw sufficient condition before normalization. Rescaling $(\xi,\varsigma)$ by $L_A$ yields 
\eqref{eq:corrected_condition}. When $H_K(\boldsymbol{\rho}) = 0$, we recover 
$\xi\varsigma < 1$ \cite{kang2023time}. \hfill$\square$

\subsection{Complexity Analysis of Wasserstein Computation}
\label{app:wasserstein}

In the 1D queue state space $\Omega = [0, Q_{\max}]$ with $N_q$ discretization points, the Wasserstein-2 distance between two empirical distributions $\rho^{(k)}$ and $\rho^{(k')}$ reduces to the $L^2$ distance between their quantile functions \cite{qiao2025hmfg}:
\begin{equation}
    \begin{aligned}
    W_2^2(\rho^{(k)}, \rho^{(k')}) 
    &= \int_0^1 \left|F_k^{-1}(u) - F_{k'}^{-1}(u)\right|^2 du \\
    &\approx \frac{1}{N_q}\sum_{i=1}^{N_q}\left|q_{(i)}^{(k)} - q_{(i)}^{(k')}\right|^2,
    \end{aligned}
\end{equation}
where $F_k^{-1}$ is the quantile function of $\rho^{(k)}$ and $q_{(i)}^{(k)}$ is the $i$-th order statistic.

\textbf{Sorting cost:} Computing $N_q$ order statistics for each of $K$ types requires $O(K N_q \log N_q)$ operations.

\textbf{Distance computation:} After sorting, computing all $K(K-1)/2$ pairwise distances requires $O(K^2 N_q)$ operations.

\textbf{Total per-iteration cost:} The Wasserstein computation contributes $O(K^2 N_q \log N_q)$ per iteration, which is dominated by the $O(K N_q N_t)$ cost of the FPK/HJB updates (Proposition~\ref{prop:complexity}) since $K \ll N_q$ in practice.

\textbf{Independence from vehicle count $N$:} Crucially, none of the above operations involve $N$ directly. The mean field $\rho^{(k)}$ is maintained as a probability distribution over the $N_q$-point state space grid, not as a sum of $N_k$ empirical delta functions. This confirms that Algorithm~\ref{alg:hmfg_pdhg} achieves computation time \emph{independent of the number of vehicles}, analogous to the homogeneous G-prox PDHG of \cite{kang2023time,wang2024mfg_iscc}. \hfill$\square$

\subsection{Proofs of Propositions and Corollaries}
\label{app:prop_cor_proofs}

\textbf{Corollary~\ref{cor:scaling_2d}.}
\label{app:cor_scaling_2d}
Substitute $\alpha = 1/4$ and $\beta = 1$ into Theorem~\ref{thm:optimal_K}, 
giving $\gamma = (1/4)/(1/4 + 1) = 1/5$ and \eqref{eq:kstar_2d}. 
The logarithmic correction follows from \cite[Theorem~1]{fournier2015rate} 
for $d = 2$, and the minimum error rate follows by substituting into 
\eqref{eq:min_error}. \hfill$\square$

\textbf{Proposition~\ref{prop:het_adjusted}.}
\label{app:prop_het_adjusted}
Substitute $C_1 = \bar{C}_1 H_\infty$ into \eqref{eq:kstar} and simplify. \hfill$\square$

\textbf{Corollary~\ref{cor:rounded}.}
\label{app:cor_rounded}
Since $K^*$ uniquely minimizes the smooth function $\mathcal{E}(N,\cdot)$, Taylor 
expansion at $K^*$ with $|\hat{K}-K^*|\le 1/2$ gives 
$\mathcal{E}(N,\hat{K})-\mathcal{E}(N,K^*) \le \frac{1}{8}\mathcal{E}''(N,K^*)$. 
Because $\mathcal{E}''(N,K^*)$ is polynomial in $N$ while $\mathcal{E}(N,K^*)$ 
decays as \eqref{eq:min_error}, the relative gap is $o(1)$. \hfill$\square$

\textbf{Corollary~\ref{cor:unbalanced_explicit}.}
\label{app:cor_unbalanced_explicit}
Replace $N$ by $\lambda_{\min}N$ in the sample term of \eqref{eq:reduced_error} 
(governed by the smallest class $n_{N,K}^{\min}=\lambda_{\min}N$) and apply 
Theorem~\ref{thm:optimal_K} to the resulting two-term functional. \hfill$\square$

\textbf{Corollary~\ref{cor:leo_order_optimal}.}
\label{app:cor_leo_order_optimal}
The static optimal rate is $\Theta(N^{-\alpha\beta/(\alpha+\beta)})$ by 
\eqref{eq:min_error}. Under \eqref{eq:leo_condition} the additive LEO perturbation 
in \eqref{eq:leo_error} is asymptotically subdominant, so the leading-order 
minimizer scaling is unchanged. \hfill$\square$

\textbf{Proposition~\ref{prop:complexity}.}
\label{app:prop_complexity}
Each iteration requires $K$ parallel FPK/HJB updates of cost $O(N_q N_t)$ plus 
$K(K-1)/2$ Wasserstein distances. In 1D, $W_2$ reduces to an $L^2$ quantile 
distance computable in $O(N_q\log N_q)$, giving total per-iteration cost 
$O(K^2 N_q N_t)$, independent of $N$. \hfill$\square$

\end{document}